\documentclass[journal, final]{IEEEtran}
\usepackage{graphicx}
\usepackage{amsmath}
\usepackage{amsfonts}
\usepackage{amssymb}
\usepackage{url}

\begin{document}

\title{ Buffer-Aided Relaying with Adaptive Link Selection | Fixed and Mixed Rate Transmission}
\author{Nikola~Zlatanov   and Robert Schober 
\thanks{N. Zlatanov and R. Schober are with the Department of Electrical and Computer Engineering, University of British Columbia (UBC), Vancouver, BC, V6T 1Z4,
Canada, E-mail: zlatanov@ece.ubc.ca, rschober@ece.ubc.ca}
} 

\maketitle

\begin{abstract}
We consider a  simple network consisting of a source, a half-duplex decode-and-forward relay with a buffer, and a destination. We assume that the direct source-destination link is not available and all links undergo fading.
We propose two new buffer-aided relaying schemes with different requirements regarding the availability of channel state information at the transmitter (CSIT). In the first scheme, neither the source nor the relay have full CSIT,
and consequently, both nodes are forced to transmit with fixed rates. In contrast, in the second scheme,  the source does not have full CSIT and transmits with fixed rate but the relay has full CSIT and adapts its transmission rate 
accordingly. In the absence of delay constraints, for both fixed rate and mixed rate transmission, we  derive the throughput-optimal buffer-aided relaying protocols which select either the source or the relay for transmission based on the 
instantaneous signal-to-noise ratios (SNRs) of the source-relay and relay-destination links. In addition, for the delay constrained case, we develop buffer-aided relaying protocols that achieve a predefined average delay. 
Compared to conventional relaying protocols, which select the transmitting node according to a predefined schedule independent of the instantaneous link SNRs, the proposed buffer-aided protocols with adaptive link selection
achieve large performance gains. In particular, for fixed rate transmission, we show that the proposed protocol achieves a diversity gain of two as long as an average delay of more than three time slots can be afforded. Furthermore, 
for mixed rate transmission with an average delay of $E\{T\}$ time slots, a multiplexing gain  of $r=1-1/(2E\{T\})$ is achieved. As a by-product of the considered link adaptive protocols, we also develop a novel conventional relaying protocol
for mixed rate transmission, which yields the same multiplexing gain as the protocol with adaptive link selection. Hence, for mixed rate transmission, for sufficiently large average delays, buffer-aided half-duplex relaying
with and without adaptive link selection does not suffer from a multiplexing gain loss compared to full-duplex relaying. 
\end{abstract}


\IEEEpeerreviewmaketitle
 
\newtheorem{theorem}{Theorem}
\newtheorem{lemma}{Lemma}
\newtheorem{corollary}{Corollary}
\newtheorem{remark}{Remark}
\newtheorem{definition}{Definition}
\newtheorem{proposition}{Proposition}
\section{Introduction}
Node cooperation can introduce significant throughput and diversity gains in wireless networks. The relay channel was first investigated by van der Meulen \cite{meulen}. Later Cover and El Gamal \cite{cover} investigated the memoryless three-node relay channel 
consisting of a source, a destination, and a single full-duplex relay and  proved that cooperative systems offer throughput gains compared to non-cooperative systems. This work was later extended to systems employing a half-duplex relay in fading environments 
for the case when the relay has a predetermined schedule for reception and transmission \cite{1435648}. For the case of fixed rate transmission, the outage probability of the three-node relay network was shown to be  superior to non-relay aided transmission in \cite{erkip1,erkip2}. Subsequently, in \cite{laneman1}, a simple protocol for the three-node relay network, which requires feedback from the receiver, was shown to achieve a diversity order of two in Rayleigh fading if the direct source-destination link is available for transmission. 
These early contributions have sparked a significant interest in cooperative communication techniques which resulted in many new discoveries, e.g.,~\nocite{laneman2,nostrina,hunter2006diversity,1542409,eth,Bletsas06,4305421,1499041,4305399}\cite{laneman2}-\cite{4305399}. 
\subsection{Background and Related Work}
 In practice, half-duplex relays may be preferred as they are easier to implement than full-duplex relays. However, half-duplexing suffers from a multiplexing gain loss compared to full duplexing. To compensate for this loss, existing protocols for the wireless three-node 
 network with a half-duplex relay exploit  the direct source-destination link to achieve a  throughput gain or a diversity gain over non-relay aided transmission, e.g.,~\cite{1435648}-\cite{eth}. In practice, because of the typically large distance between source and destination, the
 direct source-destination link may be very  weak and the gains may manifest themselves only at very high signal-to-noise ratios (SNRs). However, if a source-destination link is not available, much of the gains  obtained by half-duplex relaying disappear. 
There are two reasons for this.  First, in most of the existing literature, e.g.,~\cite{1435648}-\cite{eth}, the schedule of when the source transmits and when the relay transmits is a priori fixed. Typically, the relay receives a codeword from the source in one time slot and forwards some 
information about the received codeword to the destination in the next time slot. We refer to this approach in the following as ``conventional relaying". Second, even if the relay has channel state information at the transmitter (CSIT),  it does not exploit this information for
rate adaptation, see, e.g.,~\cite{Bletsas06}. In this paper, we propose relaying protocols that select the transmitting node based on the quality of the source-relay and the relay-destination links, i.e., the schedule of transmission is not a priori fixed. For this to be possible, the
relays have to be equipped with buffers for data storage, the node performing the selection of the transmitting node requires some channel state information (CSI) of both involved links, and feedback of a few bits of information from the node performing the selection
to the transmitting node is necessary. Furthermore, we assume that if the relay has CSIT, it exploits this knowledge to adapt the transmission rate over the relay-destination channel. 

Relays with buffers have been considered in the literature before \nocite{XFTP08,Aissa11_J,5205645,5934676,5199121} \cite{XFTP08}-\cite{5199121}. In \cite{XFTP08}, the buffer at the relay is used to enable the relay to receive for a fixed number of time slots before 
retransmitting the received information in a fixed number of time slots. In \cite{Aissa11_J}, relay selection is considered and buffers enable the selection of the relay with the best source-relay channel for reception and the best relay-destination channel for transmission. 
However, in both \cite{XFTP08} and \cite{Aissa11_J}, the schedule of when the source transmits and when the relays transmit is a priori fixed. Thus, these schemes do not achieve a diversity gain compared to conventional relaying. 
Buffer-aided relaying schemes, where the schedule of when the source transmits and when the relay transmits is not a priori fixed, are considered in \cite{5205645}-\cite{5199121}. In \cite{5205645}, the authors propose a protocol for relay selection in a network 
employing multiple mobile relays with buffers. The protocol  operates in one of the following three modes: 1) If there are relay-destination links whose SNR is sufficiently high for successful transmission and the corresponding relays have packets in their buffers, a single relay is chosen to 
transmit to the destination;  else 2) if there are source-relay links with sufficiently high SNR, the source is selected for transmission; else 3) none of the nodes transmits.  Furthermore, \cite{5934676} considers a diamond cooperative network with two relays and buffering at the
relays is used only when:  1) The instantaneous SNRs of both source-relay links are smaller than some predefined threshold while the instantaneous SNR of at least one of the relay-destination links is larger than the threshold, or 2) the instantaneous SNRs of both relay-destination  
links are smaller than the threshold while the instantaneous SNR of at least one of the source-relay links  is  larger than the threshold. Moreover, the authors in \cite{5199121} introduce a relay selection scheme for a network employing multiple relays with buffers. 
In this scheme, the schedule of when a relay receives and transmits depends on the number of packets in the relay's buffer and the instantaneous SNRs of the source-relay and relay-destination links. Although the protocols proposed in \cite{5205645}-\cite{5199121} yield a throughput 
gain over conventional relaying, they were derived based on heuristics, and are thus generally not optimal as far as throughput maximization and/or outage probability minimization are concerned. Consequently, these protocols do not fully exploit the degrees 
of freedom offered by relays with buffers.

For the case of adaptive rate transmission, the maximum achievable throughput of the simple three-node relay network employing a half-duplex decode-and-forward relay with a buffer was recently derived in \cite{BA-relaying-adaptive-rate,globe11}. 
Thereby, both the source and the relay were assumed to adjust their transmission rate such that outages are avoided. However, adjusting the rate of transmission is not possible if CSIT  is not available and/or only one modulation/coding scheme is 
implemented. In these cases, the protocol proposed in \cite{BA-relaying-adaptive-rate,globe11} is not applicable. Some preliminary results on buffer-aided relaying for fixed rate transmission have been presented in \cite{globe11} and independently in \cite{globe11a}. However, although \cite{globe11,globe11a} demonstrate that the simple three-node network with one buffer-aided relay and without direct source-destination link can achieve a diversity order of two in Rayleigh fading, the protocols 
adopted in \cite{globe11,globe11a} are suboptimal. Specifically, the protocol in  \cite{globe11} employs a suboptimal decision function for link selection, and the protocol in \cite{globe11a} only considers the instantaneous link SNRs for link selection
but does not take into account the average link SNRs, which may lead to low throughputs for non-identical average link SNRs. The idea of adaptive link selection in \cite{globe11} was extended to relay selection in \cite{Krikidis}, where  a suboptimal decision function 
exploiting the instantaneous link SNRs only was employed for link selection. We note that for the case of one relay  and identical average link SNRs, the fixed rate schemes in  \cite{globe11, globe11a, Krikidis} are all identical.
Furthermore, for mixed rate transmission, where the source transmits with fixed rate but the relay can adjust its rate to the channel conditions, some preliminary results have been reported for buffer-aided relaying in \cite{iwcmc_paper}. 
Here, we extend the protocol in \cite{iwcmc_paper} to the case of power allocation and propose a new protocol for conventional mixed rate relaying with delay constraints.
\subsection{Contributions}
In this paper, we consider the simple three-node relay network with a half-duplex decode-and-forward relay, which is equipped with a buffer, and assume that the direct source-destination link is not available for transmission. We assume that both the source-relay and the relay-destination
links are affected by fading. Depending on  the availability of CSIT at the transmitting nodes (and their capability of using more than one modulation/coding scheme),  we consider two different modes of transmission for the relay network: \textit{Fixed rate transmission} and 
\textit{mixed rate transmission}.  In both modes of transmission, each codeword spans one time slot . In fixed rate transmission, the node selected for transmission (source or relay) does not have CSIT and transmits with fixed rate. In contrast, in mixed rate transmission, the relay has CSIT knowledge and exploits it to transmit with variable rate so that outages are avoided. However, the source still transmits with fixed rate to avoid the need for CSIT acquisition. 

To explore the performance limits of the proposed fixed rate and mixed rate transmission schemes, we consider first transmission without delay constraints and derive the corresponding optimal buffer-aided relaying protocols. Since in practice it is desirable to limit the transmission delay, 
we also introduce modified buffer-aided relaying protocols for delay constrained transmission. In particular, we make the following main contributions:
\begin{itemize}
\item For fixed rate and mixed rate transmission without delay constraints, we derive the optimal buffer-aided relaying protocols which maximize the achievable throughput of the considered three-node relay network employing a half-duplex relay with a buffer of infinite size.
\item For fixed rate transmission, we show that in Rayleigh fading the optimal buffer-aided relaying protocol with adaptive link selection achieves a diversity gain of two and a diversity-multiplexing tradeoff of $DM(r)=2(1-2r)$,  where $r$ denotes the multiplexing gain.
\item For mixed rate transmission, we show that a multiplexing gain of one can be achieved with buffer-aided relaying with and without adaptive link selection implying that there is no multiplexing gain loss compared to ideal full-duplex relaying. 
\item For fixed rate and mixed rate transmission with delay constraints, in order to control the average delay, we introduce appropriate modifications to the buffer-aided relaying protocols for the delay unconstrained case. Surprisingly, for fixed rate transmission, the
full diversity gain is preserved as long as the tolerable average delay exceeds three time slots.  For mixed rate transmission with an average delay of $E\{T\}$ time slots, a multiplexing gain of $r=1-1/(2E\{T\})$ is achieved. 
\end{itemize} 
\subsection{Organization}
The remainder of this paper is organized as follows. In Section \ref{sys-mod}, the system model of the considered three-node relay network is presented. In Sections \ref{sec-3} and \ref{sec-delay}, we introduce the proposed buffer-aided relaying protocols 
for delay unconstrained and delay constrained fixed rate transmission, respectively. Protocols for delay unconstrained and delay constrained mixed rate transmission are proposed and analyzed in Section \ref{sec-hybrid}. The derived analytical results and relay 
protocols are verified and illustrated with numerical examples in Section \ref{numerics}, and some conclusions are drawn in Section \ref{conclude}.
\section{System Model and Benchmark Schemes}\label{sys-mod}
We consider a three-node wireless network comprising a source $\mathcal{S}$, a half-duplex decode-and-forward relay $\mathcal{R}$, and a destination $\mathcal{D}$, cf.~Fig.~\ref{fig0}. The source can communicate with the destination only through the relay, i.e., 
there is no direct $\mathcal{S}$-$\mathcal{D}$ link. The source sends codewords  to the relay, which decodes these codewords, possibly stores the decoded information in its buffer, and eventually sends it to the destination.
We assume that time is divided into slots of equal lengths and every codeword spans one time slot. Throughout this paper, we assume that the source node has  always data to transmit. Hence, the total number of time slots, denoted by $N$, satisfies $N\to\infty$. 
Furthermore, unless specified otherwise, we assume that the buffer at the relay is not limited in size. The case of limited buffer size will be investigated in Sections \ref{sec-delay} and \ref{mixed-delay}.
\begin{figure}
\includegraphics[width=3.5in]{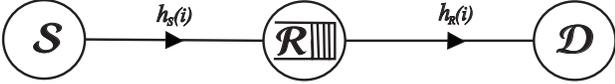}
\vspace*{-7mm}
\caption{System model for three node relay network employing a half-duplex decode-and-forward relay. The relay is equipped with a buffer to store the packets received from the source.} \label{fig0}
\end{figure}
\subsection{Channel Model}\label{chan-mod}
 In the $i$th time slot, the transmit powers of source and relay are denoted by $\mathcal{P_S}(i)$ and  $\mathcal{P_R}(i)$, respectively, and the instantaneous squared channel gains of the 
$\mathcal{S}$-$\mathcal{R}$ and $\mathcal{R}$-$\mathcal{D}$ links are denoted by $h_S(i)$ and $h_R(i)$, respectively. $h_S(i)$ and $h_R(i)$ are modeled as mutually independent, non-negative, stationary, and ergodic random processes 
with expected values $ \bar\Omega_S\triangleq E\{h_S(i)\}$ and $\bar\Omega_R\triangleq E\{h_R(i)\}$, where $E\{\cdot\}$ denotes expectation. We assume that the channel gains are constant 
during one time slot but change from one time slot to the next due to, e.g.,~the mobility of the involved nodes and/or frequency hopping. We note that for most results derived in this paper, we only require $h_S(i)$ and $h_R(i)$ to be not
fully temporally correlated, respectively. However, in some cases, we will assume that $h_S(i)$ and $h_R(i)$ are temporally uncorrelated, respectively, to facilitate the analysis.

The instantaneous SNRs of the $\mathcal{S}$-$\mathcal{R}$ and $\mathcal{R}$-$\mathcal{D}$ channels in the $i$th time slot are given by $s(i)\triangleq \gamma_S(i)  h_S(i)$ and $r(i)\triangleq \gamma_R(i)  h_R(i)$, respectively. Here, 
$\gamma_S(i)\triangleq \mathcal{P_S}(i)/\sigma_{n_R}^2$ and $\gamma_R(i)\triangleq \mathcal{P_R}(i)/\sigma_{n_D}^2$ denote the average transmit SNRs of the source and the relay, respectively, and $\sigma_{n_R}^2$ and $\sigma_{n_D}^2$ are the variances 
of the additive white Gaussian noise (AWGN) at the relay and  the destination, respectively. The average link SNRs are denoted by $\Omega_S\triangleq E\{s(i)\}$ and  $\Omega_R\triangleq E\{r(i)\}$. 

Furthermore, for concreteness, we specialize some of the derived results to Rayleigh fading. In this case, the probability density functions (pdfs) of $s(i)$ and $r(i)$ are given by $f_{s}(s)=e^{-s/\Omega_S}/\Omega_S$ 
and $f_{r}(r)=e^{-r/\Omega_R}/\Omega_R$, respectively. Similarly, the pdfs of $h_S(i)$ and $h_R(i)$ are given by $f_{h_S}(h_S)=e^{-h_S/\bar\Omega_S}/\bar\Omega_S$ and $f_{h_R}(h_R)=e^{-h_R/\bar\Omega_R}/\bar\Omega_R$, respectively.
\subsection{Link Adaptive Transmission Protocol\label{s22}}
For the proposed link adaptive transmission protocol, we assume that the relay selects which node (source or relay) transmits in a given time slot. To this end, the relay is assumed to know  the statistics of the $\mathcal{S}$-$\mathcal{R}$ and
$\mathcal{R}$-$\mathcal{D}$ channels. Since the statistics change much more slowly than the instantaneous channel gains, the overhead necessary to acquire them  is low. Furthermore, to be able to perform coherent detection, relay and destination have to acquire $h_S(i)$ and $h_R(i)$, respectively, based on pilot symbols emitted by the source and relay, respectively.  Whether or not the relay is assumed 
to have knowledge of $h_R(i)$ for adaptive link selection depends on the mode of transmission. Furthermore, depending on the mode of transmission, relay and/or destination may require knowledge of the (fixed) transmission rates, (fixed) transmit powers, 
and noise variances $\sigma_{n_R}^2$ and $\sigma_{n_D}^2$.
%
\subsubsection{Fixed Rate Transmission} 
For fixed rate transmission, neither the source nor the relay have full CSIT, i.e., source and relay do not know $h_S(i)$ and $h_R(i)$, respectively. 
Therefore, both nodes can transmit only with predetermined fixed rates $S_0$ and $R_0$, respectively, and cannot perform power allocation, i.e., the transmit powers are 
a priori fixed as $\mathcal{P_S}(i)=\mathcal{P_S}$ and $\mathcal{P_R}(i)=\mathcal{P_R}$, $\forall i$. For the relay to be able to decide which node should transmit, it requires knowledge 
of the outage states of the $\mathcal{S}$-$\mathcal{R}$ and $\mathcal{R}$-$\mathcal{D}$ links. The relay can determine whether or not the $\mathcal{S}$-$\mathcal{R}$ link is in outage 
based on $S_0$, $\mathcal{P_S}$, $\sigma_{n_R}^2$, and $h_S(i)$. The destination can do the same for the $\mathcal{R}$-$\mathcal{D}$ link based on $R_0$, $\mathcal{P_R}$, $\sigma_{n_D}^2$, 
and $h_R(i)$, and inform the relay whether or not the $\mathcal{R}$-$\mathcal{D}$ link is in outage using one bit of feedback. Based on the outage states of the $\mathcal{S}$-$\mathcal{R}$ 
and $\mathcal{R}$-$\mathcal{D}$ links in a given time slot $i$ and the statistics of both links, the relay selects the transmitting node according to the adaptive link selection protocols 
introduced in Sections \ref{sec-3} and \ref{sec-delay}, and informs the source and destination about its decision.
\subsubsection{Mixed Rate Transmission} For this mode of transmission,  we assume that the relay has full CSIT, i.e., it knows $h_R(i)$, and can therefore adjust its transmission rate and transmit power 
$\mathcal{P_R}(i)$ to avoid outages on the $\mathcal{R}$-$\mathcal{D}$ link. However, the source still  does not have CSIT and therefore has to transmit with fixed rate $S_0$ and fixed power $\mathcal{P_S}$ 
as it does not know $h_S(i)$. Similar to the fixed rate case, the relay can determine the outage state of the $\mathcal{S}$-$\mathcal{R}$ link  based on $S_0$, $\mathcal{P_S}$, $\sigma_{n_R}^2$, and $h_S(i)$. 
However, different from the fixed rate case, in the mixed rate transmission mode, the relay also has to estimate $h_R(i)$, e.g.,~based on pilot symbols emitted by the destination. Based on  the outage state of the $\mathcal{S}$-$\mathcal{R}$ link and $h_R(i)$, and on   the 
  statistics of both links, the relay selects the transmitting node according to the adaptive link selection protocols proposed in Section \ref{sec-hybrid}, and informs the source and 
destination about its decision.
 
 For both modes of transmission, the relay knows the outage state of the $\mathcal{S}$-$\mathcal{R}$ and the $\mathcal{R}$-$\mathcal{D}$ links. Hence, if the relay is selected for transmission but the 
$\mathcal{R}$-$\mathcal{D}$ link is in outage, the relay remains silent and an outage event occurs. Whereas, if the source is selected for transmission and the  $\mathcal{S}$-$\mathcal{R}$  link is in outage, 
the relay informs the source accordingly and the source remains silent, i.e., again an outage event occurs.
Once the decision regarding the transmitting node has been made, and the relay has informed the source and the destination accordingly, transmission in time slot $i$ begins.

\begin{remark}
We note that fixed rate transmission requires only two emissions of pilot symbols (by source and relay).  In contrast, mixed rate transmission requires three emissions of pilot symbols (by source, relay, and destination).
Thus, the CSI requirements and feedback overhead of the buffer-aided link selection protocols proposed in this paper are similar to those of existing relaying protocols, such as the opportunistic protocol  proposed in 
\cite{Bletsas06}. Namely, the protocol proposed in \cite{Bletsas06} requires  the relays to acquire the instantaneous CSI  of the $\mathcal{S}$-$\mathcal{R}$ and  $\mathcal{R}$-$\mathcal{D}$ links. Furthermore, a few bits of 
information are fed back from the relays to both the source and the destination.
\end{remark} 

\subsection{Queue at the Relay}
Crucial for derivation of the proposed link selection protocols is a clear understanding of the dynamics of the queue in the buffer of the relay.
In the following, for convenience, we normalize the number of bits transmitted in one time slot to the number of symbols per time slot. 
Thus, throughout the remainder of this paper,  when we refer to the number of bits, we mean the  number of bits normalized by the number of symbols in a codeword. 

If the source is selected for transmission in time slot $i$ and an outage does not occur, i.e., $\log_2\big(1+s(i)\big)\geq S_0$, it transmits with rate $S_{\mathcal{S}\mathcal{R}}(i)=S_0$.
 Hence, the relay 
receives $S_0$ data bits from the source and appends them to the queue in its buffer. The number of bits in the buffer of the relay at the end of the $i$-th time slot is denoted by $Q(i)$ and given by 
\begin{eqnarray}\label{eq2}
    Q(i)=Q(i-1)+ S_0.
\end{eqnarray}
If the source is selected for transmission but the $\mathcal{S}$-$\mathcal{R}$ link is in outage, i.e., $\log_2\big(1+s(i)\big)<S_0$, the source remains silent, i.e., $S_{\mathcal{S}\mathcal{R}}(i)=0$, and the queue in the buffer remains unchanged, i.e., $Q(i)=Q(i-1)$.

For fixed rate transmission, if the relay is selected for transmission in time slot $i$ and transmits with rate $R_0$, an outage does not occur if $\log_2\big(1+r(i)\big)\geq R_0$. In this case, the number of bits transmitted by the relay is given by
\begin{eqnarray}\label{3-lita}
    R_{\mathcal{R}\mathcal{D}}(i)=\min\{R_0,Q(i-1)\},
\end{eqnarray}
where we take into account that the maximum number of bits that can be send by the relay is limited by the number of bits in the buffer. The number of data bits remaining in the buffer at the end of time slot $i$ is given by
\begin{eqnarray}\label{eq4}
    Q(i)=Q(i-1)-R_{\mathcal{R}\mathcal{D}}(i),
\end{eqnarray}
which is always non-negative because of (\ref{3-lita}). 
If the relay is selected for transmission in  time slot $i$ but an outage occurs, i.e., $\log_2\big(1+r(i)\big)<R_0$, the relay remains silent, i.e., $R_{\mathcal{R}\mathcal{D}}(i)=0$, while the queue in the buffer remains unchanged, i.e.,  $Q(i)=Q(i-1)$.

For mixed rate transmission, the relay is able to adapt its rate to the capacity of the $\mathcal{R}$-$\mathcal{D}$ channel, $\log_2(1+r(i))$, and outages are avoided. If the relay is selected for transmission in time slot $i$, the number of bits transmitted by  
the relay is given by
\begin{eqnarray}\label{3-lit}
    R_{\mathcal{R}\mathcal{D}}(i)=\min\{\log_2(1+r(i)),Q(i-1)\}.
\end{eqnarray}
The number of data bits remaining in the buffer at the end of time slot $i$ is still given by (\ref{eq4}) where $R_{\mathcal{R}\mathcal{D}}(i)$ is now given by (\ref{3-lit}).

Furthermore, because  of the half-duplex constraint, for both fixed and mixed rate transmission, we have $R_{\mathcal{R}\mathcal{D}}(i) = 0$ and $S_{\mathcal{S}\mathcal{R}}(i)= 0$ 
if source and relay are selected for transmission in time slot $i$, respectively. 
\subsection{Link Outages and Indicator Variables}
For future reference, we introduce the binary link outage indicator variables $O_{S}(i)\in\{0,1\}$ and $O_{R}(i)\in\{0,1\}$ defined as
\begin{eqnarray}\label{eq-O_S}
O_{S}(i)\triangleq\left\{
\begin{array}{cl}
0 & \textrm{if }   s(i)<2^{S_0}-1 \\
1 & \textrm{if }   s(i)\geq 2^{S_0}-1 
\end{array} 
\right.\;
\end{eqnarray} 
and
\begin{eqnarray}\label{eq-O_R}
O_{R}(i)\triangleq\left\{
\begin{array}{cl}
0 & \textrm{if }   r(i)<2^{R_0}-1 \\
1 & \textrm{if }   r(i)\geq 2^{R_0}-1 
\end{array} 
\right.\;,
\end{eqnarray} 
respectively. In other words, $O_{S}(i)=0$ indicates that for transmission with rate $S_0$, the $\mathcal{S}$-$\mathcal{R}$ link is in outage, i.e., $\log_2(1+s(i))<S_0$, and $O_{S}(i)=1$ indicates that the transmission over the $\mathcal{S}$-$\mathcal{R}$ channel will be successful. 
Similarly, $O_{R}(i)=0$ indicates that for transmission with rate $R_0$, the $\mathcal{R}$-$\mathcal{D}$ link is in outage, i.e., $\log_2(1+r(i))<R_0$, and $O_{R}(i)=1$ means that an outage will not occur. Furthermore,
we denote the outage probabilities of the $\mathcal{S}$-$\mathcal{R}$ and $\mathcal{R}$-$\mathcal{D}$ channels as $P_S$ and $P_R$, respectively. These probabilities are defined as 
\begin{equation}
P_S\triangleq \lim_{N\to\infty}\frac{1}{N}\sum_{i=1}^N\big(1- O_S(i)\big)= {\rm Pr}\big \{s(i)<2^{S_0}-1\big\}\label{PS}
\end{equation}
and
\begin{equation}
P_R\triangleq \lim_{N\to\infty}\frac{1}{N}\sum_{i=1}^N\big(1- O_R(i)\big) = {\rm Pr}\big\{r(i)<2^{R_0}-1\big\}\label{PR},
\end{equation}
respectively.
\subsection{Performance Metrics}
In this paper, we adopt the throughput and the outage probability as performance metrics. 

Assuming the source has always data to transmit, for both fixed and mixed rate transmission, the average number of bits that arrive at the destination per time slot is given by 
\begin{eqnarray}\label{eq1}
    \tau=\lim_{N\to\infty}\frac{1}{N}\sum_{i=1}^N R_{\mathcal{R}\mathcal{D}}(i),
\end{eqnarray}
i.e., $\tau$ is the throughput of the considered communication system.

The outage probability is defined as the probability that the instantaneous channel capacity is unable to support some predetermined fixed transmission rate. In the considered system, an outage does not cause information loss since
the relay knows in advance whether or not the selected link can support the chosen transmission rate and data is only transmitted if the corresponding link is not in outage. Nevertheless, outages  still affect the achievable throughput
negatively. In fact, the outage probability can be interpreted as the  \textit{fraction of the throughput lost due to outages}. Thus, denoting the maximum throughput of a system in the absence of outages by $\tau_0$ and the throughput in the 
presence of outages by $\tau$, the outage probability, $F_{\rm out}$, can be expressed as
\begin{eqnarray}\label{OP-main}
    F_{\rm out}=1-\frac{\tau}{\tau_o}\;.
\end{eqnarray}
Note that maximizing the throughput is equivalent to minimizing the outage probability.

\subsection{Performance Benchmarks for Fixed Rate Transmission}
For fixed rate transmission, two conventional relaying schemes serve as performance benchmarks for the proposed buffer-aided relaying scheme with adaptive link selection. In contrast to the proposed scheme, the benchmark 
schemes employ a predetermined schedule for when source and relay transmit which is independent of the instantaneous link SNRs. 

In the first scheme, referred to as Conventional Relaying 1 (see also \cite{XFTP08}), the source transmits in the first $\xi N$ time slots, where $0<\xi<1$ and each codeword spans one time slot. The relay tries to decode these codewords and, if the decoding is successful, 
it stores the corresponding information bits in its buffer. In the following $(1-\xi)N$ time slots, the relay transmits the stored information bits  to the destination, transmitting one codeword per time slot. Assuming that for the benchmark schemes source and relay 
transmit codewords having the same rate, i.e., $S_0=R_0$, the throughput of Conventional Relaying 1 is obtained as  
\begin{eqnarray}\label{eq_qq-conv_1aa}
    \tau_{\rm{conv,1}}^{\rm fixed}=&&\hspace{-6mm} \lim_{N\to \infty} \frac{1}{N}  \min\left\{ \sum_{i=1}^{\xi N} R_0 O_S(i)\;\;, \sum_{i=\xi N+1}^{N}  R_0 O_R(i)
\right\}\nonumber\\
=&&\hspace{-6mm}R_0  \min\left\{ \xi(1-P_S)  \;,\; (1-\xi)(1-P_R)   \right\}  .
\end{eqnarray}
The throughput is maximized if $\xi (1-P_S) = (1-\xi )(1-P_R)$ holds or equivalently if  $\xi=(1-P_R)/(2-P_S-P_R)$. Inserting $\xi$ into (\ref{eq_qq-conv_1aa}) 
we obtain the maximized throughput as 
\begin{eqnarray}\label{eq_qq-conv_1a}
    \tau_{\rm{conv,1}}^{\rm fixed}=R_0\frac{(1-P_S)(1-P_R)}{2-P_S-P_R}   .
\end{eqnarray}
The maximum throughput in the absence of outages is  $\tau_0=R_0/2$, hence using  (\ref{OP-main}), the corresponding outage probability is obtained as 
\begin{eqnarray}\label{op_conv_1}
    F_{{\rm out, conv,1}}^{\rm fixed}=1- 2\frac{(1-P_S)(1-P_R)}{2-P_S-P_R}.
\end{eqnarray} 

In the second scheme, referred to as Conventional Relaying 2, in the first time slot, the source transmits one codeword and the relay receives and tries to
decode the codeword. If the decoding is successful, in the second time slot, the relay retransmits the information to the destination, otherwise it remains silent.  
The throughput of Conventional Relaying 2 is obtained as
\begin{eqnarray}\label{trup_conv_2}
    \tau_{{\rm conv,2}}^{\rm fixed}=&&\hspace{-6mm}\lim_{N\to \infty} \frac{1}{N} \sum_{i=1}^{N/2} R_0  O_S(2i-1) O_R(2i)\nonumber\\
=&&\hspace{-6mm}\frac{R_0}{2}(1-P_S)(1-P_R).
\end{eqnarray}  
Based on (\ref{OP-main}) the corresponding outage probability is given by 
\begin{eqnarray}\label{op_conv_2}
    F_{{\rm out, conv,2}}^{\rm fixed}= 1-(1-P_S)(1-P_R).
\end{eqnarray} 
We note that $\tau_{\rm{conv,1}}^{\rm fixed}\geq \tau_{\rm{conv,2}}^{\rm fixed}$ ($F_{{\rm out, conv,1}}^{\rm fixed}\leq F_{{\rm out, conv,2}}^{\rm fixed}$) always holds. However, in order for Conventional Relaying 1 to realize this gain, an infinite delay is required, whereas Conventional Relaying 2 requires a delay of only one time slot.

For the special case of Rayleigh fading, we obtain from (\ref{PS}) and (\ref{PR}) $P_S=1-e^{-\frac{2^{R_0}-1}{\Omega_S}}$ and $P_R=1-e^{-\frac{2^{R_0}-1}{\Omega_R}}$, respectively. The corresponding throughputs and outage probabilities for 
Conventional Relaying 1 and 2 can be obtained by applying these results in (\ref{eq_qq-conv_1a})-(\ref{op_conv_2}). In particular, in the high SNR regime, when $\gamma_S=\gamma_R=\gamma \to \infty$,  we obtain $ \tau_{{\rm conv,1}}^{\rm fixed}\to R_0/2$,
$ \tau_{{\rm conv,2}}^{\rm fixed}\to R_0/2$,  and
\begin{eqnarray}
     F_{{\rm out, conv,1}}^{\rm fixed}&\to&\frac{2^{R_0}-1}{2}\frac{\bar\Omega_S+\bar\Omega_R}{\bar\Omega_S \bar\Omega_R} \frac{1}{\gamma},\quad\textrm{as }\gamma\to \infty,\quad
     \label{eq16}\\
     F_{{\rm out, conv,2}}^{\rm fixed} &\to&(2^{R_0}-1)\frac{\bar\Omega_S+\bar\Omega_R}{\bar\Omega_S \bar\Omega_R} \frac{1}{\gamma} ,\quad\textrm{as }\gamma\to \infty.\quad
     \label{eq17}
\end{eqnarray}
Hence, for fixed rate transmission, the diversity gain of Conventional Relaying 1 and 2 is one as expected.
\subsection{Performance Benchmarks for Mixed Rate Transmission}\label{ggas}
We also provide two performance benchmarks with a priori fixed link selection schedule for mixed rate transmission. The two benchmark protocols are analogous to the corresponding protocols in the fixed rate case. Thus, for Conventional Relaying 1,
the source transmits in the first $\xi N$ time slots with fixed rate $S_0$ and the relay transmits in the remaining $(1-\xi) N$ time slots with rate $R(i) = \log_2(1+r(i))$. Thus, the throughput is given by
\begin{eqnarray}\label{eq_qq-conv_1mix}
    \tau_{\rm{conv,1}}^{\rm mixed}&&\hspace{-7mm}= \hspace{-1mm}\lim_{N\to \infty} \frac{1}{N}  \min\left\{\hspace{-1mm} \sum_{i=1}^{\xi N} S_0 O_S(i), \hspace{-1mm}\sum_{i=\xi N+1}^{ N} \hspace{-2mm} \log_2(1+r(i))
\hspace{-1mm}\right\}\;\;\nonumber\\
&&\hspace{-7mm}=  \min\left\{ \xi(1\hspace{-0.5mm}-\hspace{-0.5mm}P_S)S_0  , (1-\xi) E\{\log_2(1+r(i))\}   \right\}\hspace{-1mm}.
\end{eqnarray}
The throughput is maximized if $\xi$ satisfies
\begin{equation}\label{cond-1-mixed_conv}
    \xi{S_0 (1-P_S)} = (1-\xi) E\{\log_2(1+r(i))\}\;.
\end{equation}
From (\ref{cond-1-mixed_conv}), we obtain $\xi$ as
\begin{eqnarray} 
    \xi= \frac{  E\{\log_2(1+r(i))\} }{S_0 (1-P_S)+ E\{\log_2(1+r(i))\} }.
\end{eqnarray}
Inserting $\xi$ into (\ref{eq_qq-conv_1mix}) leads to the throughput of mixed rate transmission under the Conventional Relaying 1 protocol
\begin{eqnarray}\label{eq_qq-1}
    \tau_{\rm{conv,1}}^{\rm mixed}= \frac{ S_0 (1-P_S) E\{\log_2(1+r(i))\}}{S_0 (1-P_S)+ E\{\log_2(1+r(i))\} }.
\end{eqnarray}
Assuming Rayleigh fading links $E\{\log_2(1+r(i))\}$ is obtained as
\begin{eqnarray}\label{t-conv-mixed-2}
    E\{\log_2(1+r(i))\}=
  \frac{e^{1/\Omega_R}}{\ln(2)}E_1\left(\frac{1}{\Omega_R}\right) 
\end{eqnarray}
for fixed transmit powers, where $E_1(x)=\int_x^\infty e^{-t}/t\,dt$, $x>0$, denotes the exponential integral function. If adaptive power allocation is employed, $E\{\log_2(1+r(i))\}$ becomes
\begin{eqnarray}\label{t-conv-mixed-2pa}
      E\{\log_2(1+r(i))\}=\frac{1}{\ln(2)}E_1\left(\frac{\lambda_c}{\bar \Omega_R}\right) ,
\end{eqnarray}
where $\lambda_c$ is found from the power constraint
\begin{eqnarray}
    (1-P_S)\gamma_S+\int_{\lambda_c}^\infty \left(\frac{1}{\lambda_c}-\frac{1}{h_R}\right) f_{h_R}(h_R) dh_R=2\Gamma.
\end{eqnarray}
Here, $\Gamma$ denotes the average transmit power in one time slot. In the high SNR regime, where $\gamma_S=\gamma_R=\gamma\to\infty$,  $E\{\log_2(1+r(i))\}\gg S_0(1-P_S)$ holds. Thus, the throughput in (\ref{eq_qq-1}) converges to
\begin{eqnarray}\label{eq_qq-1_high}
    \tau_{\rm{conv,1}}^{\rm mixed}\to  S_0\;, \quad\textrm{as }\gamma\to\infty\;,
\end{eqnarray}
which leads to the interesting conclusion that mixed rate transmission achieves a multiplexing rate of one even if suboptimal conventional relaying is used.

For Conventional Relaying 2, the performance of mixed rate transmission is identical to that of fixed rate transmission. Since the relay does not employ a buffer for Conventional Relaying 2, even with mixed rate transmission, the relay can only transmit
successfully all of the received information if $S_0\leq\log_2(1+r(i))$ and has to remain silent otherwise.
\section{Fixed Rate Transmission Without Delay Constraints}\label{sec-3}
In this section, we investigate buffer-aided relaying with adaptive link selection for fixed rate transmission without delay constraints, i.e., the transmission rates of the source and the relay are fixed. We derive the optimal link selection protocol and
analyze the corresponding throughput and outage probability. The obtained results constitute performance upper bounds for fixed rate transmission with delay constraints, which will be considered in Section \ref{sec-delay}.
\subsection{Problem Formulation}\label{sec_pf}
First, we introduce the binary link selection variable $d_i\in\{0,1\}$. Here, $d_i=1$ indicates that the $\mathcal{R}$-$\mathcal{D}$ link is selected for transmission in time slot $i$, i.e., the relay transmits and the destination receives. 
Similarly, if $d_i=0$, the $\mathcal{S}$-$\mathcal{R}$ link is selected for transmission in time slot $i$, i.e., the source transmits and the relay receives. 

Based on the definitions of $O_{S}(i)$, $O_{R}(i)$, and $d_i$, the number of bits sent from the source to the relay and from the relay to the destination in time slot $i$ can be written in compact form as
\begin{eqnarray}\label{s_2}
    S_{\mathcal{S}\mathcal{R}}(i)= (1-d_i) O_{S}(i) S_0
\end{eqnarray}
and
\begin{eqnarray}\label{s_1}
    R_{\mathcal{R}\mathcal{D}}(i)= d_i O_{R}(i) \min\{R_0,Q(i-1)\}, 
\end{eqnarray}
respectively. Consequently, the throughput in (\ref{eq1}) can be rewritten as
\begin{eqnarray}\label{trup}
    \tau=\lim_{N\to\infty}\frac{1}{N}\sum_{i=1}^N d_i O_{R}(i)  \min\{R_0,Q(i-1)\}.
\end{eqnarray}
In the following, we maximize the throughput by optimizing the link selection variable $d_i$, which represents the only degree of freedom in the considered problem. In particular, as already mentioned in Section \ref{s22}, since both transmitting nodes do not 
have the full CSI of their respective transmit channels, power allocation is not possible and we assume fixed transmit powers $\mathcal{P_S}(i)=\mathcal{P_S}$ and $\mathcal{P_R}(i)=\mathcal{P_R}$, $\forall i$.
\subsection{Throughput Maximization}
Let us first define the average arrival rate of bits per slot into the queue of the buffer, $A$, and the average departure rate of bits per slot out of the queue of the buffer, $D$, as \cite{2}
\begin{eqnarray}\label{arr-rate}
    A\triangleq\lim_{N\to\infty}\frac{1}{N}\sum_{i=1}^N (1-d_i) O_{S}(i)  S_0 \label{AD}
\end{eqnarray}
and
\begin{eqnarray}
 D\triangleq\lim_{N\to\infty}\frac{1}{N}\sum_{i=1}^N d_i  O_{R}(i) \min\{R_0,Q(i-1)\},\label{D}
\end{eqnarray}
respectively. We note that the departure rate of the queue is equal to the throughput. The queue is said to be an absorbing queue if $A>D=\tau$, in which case a fraction of the information sent by the source is trapped in the buffer and can never be extracted from it. The following theorem provides a useful condition for the optimal policy which  maximizes the throughput.

\begin{theorem}\label{theorem1}
The link selection  policy that maximizes the throughput of the considered buffer-aided relaying system can be found in the set of link selection policies that satisfy  
\begin{eqnarray}\label{cond-1}
    \lim_{N\to\infty}\frac{1}{N}\sum_{i=1}^N (1-d_i) O_S(i)  S_0 \hspace{-0.5mm} =\hspace{-0.8mm} \lim_{N\to\infty}\hspace{-1mm}\frac{1}{N}\sum_{i=1}^N d_i O_R(i)  R_0,
\end{eqnarray}
and the throughput is given by the right (and left) hand side of (\ref{cond-1}). If (\ref{cond-1}) holds, the queue is non-absorbing but is at the edge of absorption, i.e., a small increase of the arrival rate will lead to an absorbing queue. 
\end{theorem}
\begin{proof}
Please refer to Appendix \ref{app_for_t_1}.
\end{proof}

\begin{remark}\label{remark_new_0}
A queue that meets condition (\ref{cond-1}) is referred to as a critical queue \cite{loynes1962stability}. Critical queues may be stable, substable, or unstable.  For the optimal link selection policy in Theorem~\ref{theorem1}, the queue is non-absorbing  hence   leading to a stable queue.
\end{remark}

\begin{remark}\label{remark_new_1}
The $\min(\cdot)$ function in (\ref{trup}) is absent in the throughput in (\ref{cond-1}), which is crucial for finding a tractable analytical expression for the optimal link selection policy. In particular, as shown in Appendix A, condition 
(\ref{cond-1}) automatically ensures that for $N\to\infty$, 
$$\tau= \frac{1}{N} \sum_{i=1}^N d_i O_R(i) \min\{R_0,Q(i-1)\}=\frac{1}{N} \sum_{i=1}^N d_i O_R(i) R_0$$ is valid, i.e., the impact of event $R_0>Q(i-1)$, $i=1,\ldots,N$, 
is negligible. Hence, for the optimal link selection policy, the queue is non-absorbing but is almost always filled to such a level that the number of bits in the queue exceed the number of bits that can be transmitted 
over the $\mathcal{R}$-$\mathcal{D}$ channel, i.e., the buffer is practically always fully backlogged. This result is intuitively pleasing. Namely, if the queue would be unstable, it would absorb bits and the throughput could be
improved by having the relay transmit more frequently. On the other hand, if the queue was not (practically) fully backlogged, the effect of the event $R_0>Q(i-1)$ would not be negligible and the system would loose out on 
transmission opportunities because of an insufficient number of bits in the buffer.
\end{remark}

\begin{remark}\label{remark_new_2}
We note that Theorem 1 is only valid for $N\to\infty$ where transient effects resulting from filling the buffer at the beginning of transmission and emptying it at the end of transmission are negligible. For (small) finite $N$,
these effects are not negligible and the derivation of the optimal link selection policy is more complicated.
\end{remark}

According to Theorem~\ref{theorem1},  in order to maximize the throughput, we have to search for the optimal policy only in the set of policies that satisfy (\ref{cond-1}). 
Therefore, the search for the optimal policy can be formulated as an optimization problem,  which for $N\to\infty$ has the following form
\begin{eqnarray}\label{MPR1}
\begin{array}{ll}
 {\underset{d_i}{\rm{Maximize: }}}&\frac{1}{N}\sum_{i=1}^N d_i O_R(i) R_0\\
{\rm{Subject\;\; to: }} &{\rm C1:}\, \frac{1}{N}\sum_{i=1}^N (1-d_i) O_S(i) S_0\\
&\qquad=\frac{1}{N}\sum_{i=1}^N d_i O_R(i) R_0\\
  &{\rm C2:} \, d_i(1-d_i)=0,\quad \forall i\\
\end{array}
\end{eqnarray}
where constraint C1 ensures that the search for the optimal policy is conducted only among those policies that satisfy (\ref{cond-1}) and C2 ensures that $d_i\in\{0,1\}$. 
We note that C1 and C2 do not exclude the case that the relay is chosen for transmission if $R_0>Q(i-1)$. However, as explained in Remark \ref{remark_new_1}, C1 
ensures that the influence of event $R_0>Q(i-1)$ is negligible. Therefore, an additional constraint dealing with this event is not required.

Before we solve problem (\ref{MPR1}), we note that, as will be shown in the following, the optimal link selection policy may require a coin flip. For this purpose, 
we introduce the set of possible outcomes of the coin flip, $\mathcal{C}\in\{0,1\}$, and denote the probabilities of the outcomes by $P_C={\rm Pr}\{\mathcal{C}=1\}$  and ${\rm Pr}\{\mathcal{C}=0\}=1-P_C$, respectively. 
Now, we are ready to provide the solution of (\ref{MPR1}), which constitutes the optimal link selection policy maximizing the throughput. This is conveyed in the following theorem.

\begin{theorem}\label{theorem2} 
For the optimal link selection policy maximizing the throughput of  the considered buffer-aided relaying system for fixed rate transmission, three mutually exclusive cases can be distinguished depending on the values of $P_S$ and $P_R$:\\
\textbf{Case 1:}
\begin{eqnarray}\label{cond-PS-PR-1}
   P_S\leq \frac{S_0}{S_0+R_0(1-P_R)} &\textrm{AND}& P_R\leq \frac{R_0}{R_0+S_0(1-P_S)}.\nonumber\\
\end{eqnarray}
In this case, the optimal link selection policy is given by
\begin{eqnarray}\label{sol-d-1}
d_i=\left\{\hspace{-1.4mm}
\begin{array}{cl}
0 & \textrm{if }   O_S(i)=1 \textrm{ AND } O_R(i)=0 \\
1 & \textrm{if }   O_S(i)=0 \textrm{ AND } O_R(i)=1\\
0 & \textrm{if }   O_S(i)=1 \textrm{ AND } O_R(i)=1  \textrm{ AND } \mathcal{C}=0 \\
1 & \textrm{if }   O_S(i)=1 \textrm{ AND } O_R(i)=1  \textrm{ AND }\mathcal{C}=1 \\
\varepsilon & \textrm{if }   O_S(i)=0 \textrm{ AND } O_R(i)=0
\end{array} 
\right.
\end{eqnarray} 
where $\varepsilon$ can be set to $0$ or $1$ as neither the source nor the relay will transmit because both links are in outage. On the other hand, if both links are not in outage, i.e., $O_S(i)=1$ and $O_R(i)=1$, the coin flip decides which node transmits and
the probability of $\mathcal{C}=1$ is given by 
\begin{equation}\label{find-P_C-1}
  P_C =\frac{S_0 (1-P_S)-(1-P_R)P_S R_0 }{(1-P_S)(1-P_R)(S_0+R_0)}.
\end{equation}
Based on (\ref{sol-d-1}), the maximum throughput is obtained as
\begin{eqnarray}\label{max-tau-1}
 \tau=\frac{S_0 R_0}{S_0+R_0}(1-P_S P_R).
\end{eqnarray}
\textbf{Case 2:}
\begin{eqnarray}\label{cond-PS-PR-2}
   P_R> \frac{R_0}{R_0+S_0(1-P_S)}
\end{eqnarray}
In this case, the optimal link selection policy is characterized by 
\begin{eqnarray}\label{sol-d-2}
d_i=\left\{\hspace{-1.4mm}
\begin{array}{cl}
0 & \textrm{if }  O_S(i)=1 \textrm{ AND } O_R(i)=0   \textrm{ AND } \mathcal{C}=0 \\
1 & \textrm{if } O_S(i)=1 \textrm{ AND } O_R(i)=0  \textrm{ AND } \mathcal{C}=1  \\
1 & \textrm{if }   O_S(i)=0 \textrm{ AND } O_R(i)=1\\
1 & \textrm{if }  O_S(i)=1 \textrm{ AND } O_R(i)=1 \\
\varepsilon & \textrm{if }   O_S(i)=0 \textrm{ AND } O_R(i)=0
\end{array} 
\right.
\end{eqnarray} 
The probability of outcome $\mathcal{C}=1$ of the coin flip is given by
\begin{equation}\label{find-P_C-2}
  P_C =\frac{S_0 (1-P_S)P_R-(1-P_R) R_0 }{(1-P_S)P_R S_0},
\end{equation}
and the maximum throughput can be obtained as
\begin{eqnarray}\label{max-tau-2}
 \tau=R_0(1-P_R).
\end{eqnarray}
\textbf{Case 3:} 
\begin{eqnarray}\label{cond-PS-PR-3}
   P_S>\frac{S_0}{S_0+R_0(1-P_R)}.
\end{eqnarray}
In this case, the link selection policy that maximizes the throughput is given by
\begin{eqnarray}\label{sol-d-3}
d_i=\left\{\hspace{-1.4mm}
\begin{array}{cl}
0 & \textrm{if }  O_S(i)=1 \textrm{ AND } O_R(i)=0  \\
0 & \textrm{if }   O_S(i)=0 \textrm{ AND } O_R(i)=1  \textrm{ AND }   \mathcal{C}=0\\
1 & \textrm{if }   O_S(i)=0 \textrm{ AND } O_R(i)=1  \textrm{ AND } \mathcal{C}=1 \\
0 & \textrm{if }   O_S(i)=1 \textrm{ AND } O_R(i)=1 \\
\varepsilon & \textrm{if }    O_S(i)=0 \textrm{ AND } O_R(i)=0
\end{array} 
\right.
\end{eqnarray} 
The probability of $\mathcal{C}=1$ is given by
\begin{equation}\label{find-P_C-3}
  P_C=\frac{S_0 (1-P_S) }{R_0 (1-P_R)P_S},
\end{equation}
and the maximum throughput is
\begin{eqnarray}\label{max-tau-3}
 \tau=S_0(1-P_S).
\end{eqnarray}
\end{theorem}
\begin{proof}
Please refer to  Appendix \ref{app_A}.
\end{proof}

\begin{remark}
We note that in the second line of (\ref{sol-d-2}), we set $d_i=1$ although the $\mathcal{R}$-$\mathcal{D}$ link is in outage ($O_R(i)=0$) while the $\mathcal{S}$-$\mathcal{R}$ link is not in outage ($O_S(i)=1$). In other words, in this case, neither node transmits although 
the source node could successfully transmit.  However, if the source node transmitted in this situation,  the  queue at the relay would become an absorbing queue.  Similarly, in the second line of (\ref{sol-d-3}), we set $d_i=0$ although the $\mathcal{S}$-$\mathcal{R}$ link 
is in outage. Again, neither node transmits in order to ensure that  condition (\ref{cond-1}) is met. However, in this case, the exact same throughput as in (\ref{max-tau-3}) can be achieved with a simpler and more practical link selection policy than that in (\ref{sol-d-3}). 
This is addressed in the following lemma.
\end{remark}

\begin{lemma}
The throughput achieved by the link selection policy in (\ref{sol-d-3}) can also be achieved with the following simpler link selection policy.\\
If 
\begin{eqnarray}\label{cond-PS-PR-3a}
   P_S>\frac{S_0}{S_0+R_0(1-P_R)}\;,
\end{eqnarray}
a link selection policy  maximizing the throughput is given by
\begin{eqnarray}\label{sol-d-3a}
d_i=\left\{
\begin{array}{cl}
0 & \textrm{if }  O_S(i)=1  \\
1 & \textrm{if }   O_S(i)=0
\end{array} 
\right.\;,
\end{eqnarray} 
and the maximum throughput is
\begin{eqnarray}\label{max-tau-3a}
 \tau=S_0(1-P_S).
\end{eqnarray}
\end{lemma}
\begin{proof}
The policy given by (\ref{sol-d-3}) has the same average arrival rate as  policy (\ref{sol-d-3a}) since for both policies the source always transmits when $O_S(i)=1$. Therefore, since for both policies the queue is non-absorbing, by the law of conservation of flow, their throughputs 
are identical to their arrival rates. Thus, both policies achieve identical throughputs.
\end{proof}

\begin{remark}\label{remark_duplex}
Note that when $P_R> R_0/(R_0+S_0(1-P_S))$ ($P_S> S_0/(S_0+R_0(1-P_R))$) holds, the throughput is given by (\ref{max-tau-2}) ((\ref{max-tau-3})), which is identical to the maximal throughput that can be obtained in a point-to-point communication 
between relay and destination (source and relay).  Therefore, when $P_R> R_0/(R_0+S_0(1-P_S))$ ($P_S> S_0/(S_0+R_0(1-P_R))$) holds, as far as the achievable throughput is concerned, the three-node half-duplex  relay channel  is equivalent to 
the two-node $\mathcal{R}$-$\mathcal{D}$ ($\mathcal{S}$-$\mathcal{R}$) channel. 
\end{remark}

For comparison, we also provide the maximum throughput in the absence of outages $\tau_0$. The throughput in the absence of outages, $\tau_0$, can be obtained by setting $O_S(i)=O_R(i)=1$, $\forall i$, which is 
equivalent to setting $P_S=P_R=0$ in Theorem \ref{theorem2}. Then, Case 1 in Theorem \ref{theorem2} always holds  and the optimal link selection policy is
\begin{eqnarray}\label{sol-d_no_outage}
d_i=\left\{
\begin{array}{cl}
0 & \textrm{if }    \mathcal{C}=0 \\
1 & \textrm{if }   \mathcal{C}=1 \\
\end{array} 
\right.
\end{eqnarray} 
where the probability of $\mathcal{C}=1$ is given by 
\begin{equation}\label{find-P_C-1a}
  P_C =\frac{S_0}{S_0+R_0}.
\end{equation}
Based on (\ref{sol-d_no_outage}), the maximum throughput in the absence of outages is
\begin{eqnarray}\label{t_no_outage}
 \tau_0=\frac{S_0 R_0}{S_0+R_0} .
\end{eqnarray}
The throughput loss caused by outages can be observed by comparing (\ref{max-tau-1}), (\ref{max-tau-2}), and (\ref{max-tau-3}) with (\ref{t_no_outage}). 

We now provide the outage probability of the proposed buffer-aided relaying scheme with adaptive link selection.

\begin{lemma}\label{lemma-OP-non-delay}
The outage probability of the system considered in Theorem \ref{theorem2}  is given by
\begin{eqnarray}
     F_{\rm out}=\left\{
\begin{array}{cl}
P_R-(1-P_R)R_0/S_0\;, & \textrm{if }  P_R> \frac{R_0}{R_0+S_0(1-P_S)}\\
P_S-(1-P_S)S_0/R_0\;, & \textrm{if }  P_S> \frac{S_0}{S_0+R_0(1-P_R)} \\
P_S P_R\;, & \textrm{otherwise} .
\end{array}
\right. \nonumber
\end{eqnarray}
\vspace{-7mm}
\begin{eqnarray}
    \label{OP-non-delay}
\end{eqnarray}
\end{lemma}
\begin{proof}
Please refer to Appendix \ref{app_B}.
\end{proof}
\begin{remark}
In the proof of Lemma~\ref{lemma-OP-non-delay} given  in Appendix \ref{app_B}, it is  shown that an outage event happens when neither the source nor the relay transmit in a  time slot, i.e., the number of silent slots is identical to the number of outage events.
\end{remark}

In the high SNR regime, when the outage probabilities of both involved links are small, the expressions for the throughput and the outage probability can be simplified to obtain further insight 
into the performance of  buffer-aided relaying. This is addressed in the following lemma.
\begin{lemma}\label{lemma_1}
In the high SNR regime, $\gamma_S=\gamma_R=\gamma\to \infty$,  the throughput and the outage probability of the buffer-aided relaying system considered in Theorem \ref{theorem2} converge to 
\begin{eqnarray}
    \tau&\to&\tau_0= \frac{S_0 R_0}{S_0+R_0},\quad \textrm{as }\gamma\to\infty \;,\label{max-tau-1-high-snr}\\
     F_{\rm out}&=&  
P_S P_R . \label{prob-eq} 
\end{eqnarray}
\end{lemma}
\begin{proof}
In the high SNR regime, we have $P_S\to 0$ and  $P_R\to 0$. Thus, condition (\ref{cond-PS-PR-1}) always holds and therefore $F_{\rm out}$ is given by (\ref{prob-eq}).  Furthermore, as $P_S\to 0$ and  $P_R\to 0$, (\ref{max-tau-1}) simplifies to
(\ref{max-tau-1-high-snr}).
\end{proof}
\subsection{Performance in Rayleigh Fading}\label{sec-R}
For concreteness, we assume in this subsection that both links of the considered three-node relay system are Rayleigh fading. We examine the diversity order and the diversity-multiplexing trade-off.
\begin{lemma}\label{lemma_3}
For the special case of Rayleigh fading links,  the buffer-aided relaying system  considered in Theorem \ref{theorem2} achieves a diversity gain of two, i.e., in the high SNR regime, when $\gamma_S=\gamma_R=\gamma\to \infty$,  
the outage probability, $F_{\rm out}$,  decays on a log-log  scale with slope $-2$ as a function of the transmit SNR $\gamma$, and is given by
\begin{eqnarray}\label{prob-eq-1-high-snr}
      F_{\rm out}\to\frac{2^{S_0}-1}{\bar\Omega_S} \frac{2^{R_0}-1}{\bar\Omega_R}\frac{1}{\gamma^2}, \quad \textrm{as }\gamma\to\infty.
\end{eqnarray}
Furthermore,  the considered buffer-aided relaying system achieves a diversity-multiplexing trade-off, $DM(r)$, of
\begin{eqnarray}
    DM(r)=2(1-2 r),\quad 0<r<1/2.
\end{eqnarray}
\end{lemma}
\begin{proof}
Please refer to Appendix \ref{app_F}.
\end{proof}
\begin{remark}
We recall that, for fixed rate transmission, both considered conventional relaying schemes without adaptive link selection achieved only a diversity gain of one, cf.~(\ref{eq16}), (\ref{eq17}), despite the fact that Conventional Relaying 1 also entails an infinite delay. Thus,
we expect large gains in terms of outage probability of the proposed buffer-aided relaying protocol with adaptive link selection compared to conventional relaying.
\end{remark}

The performance of the considered system can be further improved by optimizing the transmission rates $R_0$ and $S_0$ based on the channel statistics. For Rayleigh fading with given $\bar\Omega_S$ and $\bar\Omega_R$, we can optimize $R_0$ and $S_0$
for minimization of the outage probability. This is addressed in the following lemma.

\begin{lemma}\label{lemma4}
Assuming Rayleigh fading, the optimal transmission rates $S_0$ and $R_0$ that minimize the outage probability in the high SNR regime, while maintaining a throughput of $\tau_0$, are given by $R_0=S_0=2\tau_0$.
\end{lemma}
\begin{proof}
The throughput in the high SNR regime is given by (\ref{max-tau-1-high-snr}), which can be rewritten as $R_0=S_0 \tau_0 /(S_0-\tau_0)$. Inserting this into the asymptotic expression for $F_{\rm out}$ in (\ref{prob-eq-1-high-snr}) and minimizing it with respect 
to $S_0$ yields $S_0=R_0=2\tau_0$.
\end{proof}
\begin{remark}
For Rayleigh fading, although in the low SNR regime, the optimal $S_0$ and $R_0$ can be nonidentical, in the high SNR regime, independent of the values of $\bar\Omega_S$ and $\bar\Omega_R$, the minimum $F_{\rm out}$ is obtained for identical transmission
rates for both links. Furthermore, in the high SNR regime,  when $\gamma_S=\gamma_R\to\infty$,  for $S_0=R_0$, the coin flip probability $P_C$ converges to $P_C={\rm Pr}\{\mathcal{C}=1\}={\rm Pr}\{\mathcal{C}=0\}\to 1/2$.
\end{remark}
\section{Fixed Rate Transmission With Delay Constraints}\label{sec-delay}
The protocol proposed in Section~\ref{sec-3} does not impose any constraint on the delay that a transmitted bit experiences. However, in practice, most communication services require delay constraints. Therefore, in this section, we modify the buffer-aided
relaying protocol derived in the previous section to account for constraints on the average delay. Furthermore, we analyze the effect of the applied modification on the throughput and the outage probability.   For simplicity, throughout this section, we assume $S_0=R_0$. We note that the link selection protocols proposed in Section \ref{s4B} are also applicable to the case of $S_0\ne R_0$. However, since for $S_0\ne R_0$ the packets transmitted by the source do not contain the same number of bits as 
the packets transmitted by the relay, the Markov chain based throughput and delay analyses in Sections \ref{s4C} and \ref{s4D} would be more complicated. Since we found in the previous section that, for high SNR, identical source and relay transmission rates minimize the outage probability, we avoid these additional complications here and concentrate on the case $S_0=R_0$. Furthermore,  to facilitate our analysis, throughout this section, we assume temporally uncorrelated fading.  
\subsection{Preliminaries}
We define the delay of a bit as the time interval from its transmission by the source to its reception at the destination. Thus, assuming that the propagation delays in the $\mathcal{S}$-$\mathcal{R}$  and $\mathcal{R}$-$\mathcal{D}$  links are negligible,
the delay of a bit is identical to the time that the bit is held in the buffer. As a consequence, we can use Little's law \cite{little} and express the average delay as
\begin{eqnarray}\label{delay-main}
    E\{T\}=\frac{E\{Q\}}{A},
\end{eqnarray}
where $E\{Q\}=\lim_{N\to \infty}\sum_{i=1}^N Q(i)/N$ is the average length of the queue in the buffer of the relay and $A$ is the arrival rate in bits/slot into the queue as defined in (\ref{arr-rate}).
Since $E\{Q\}$ is given in bits and $A$ is given in bits/slot, the average delay $E\{T\}$ is given in time slots. From (\ref{delay-main}), we observe that the delay can be controlled via the queue size. 
%
%
%
%
\subsection{Link Selection Protocol for Delay Limited Transmission\label{s4B}}
As mentioned before, we modify the optimal link selection protocol derived in Section~\ref{sec-3} in order to limit the average delay. However, depending on the targeted average delay, somewhat different modifications are necessary, since
it is not possible to achieve any desired delay with one protocol. Hence, three different link selection protocols are introduced in the following proposition. 
\begin{proposition}\label{proposition1}
For fixed rate transmission with delay constraint, depending on the targeted average delay $E\{T\}$ and the outage probabilities $P_S$ and $P_R$, we propose the following policies:\\
\textbf{Case 1:}
If  $P_R<1/(2-P_S)$ and the required delay $E\{T\}$ satisfies
\begin{eqnarray}\label{rr-1}
    E\{T\}>  \frac{1}{1-P_R \left(2-P_S\right)} +\frac{2 \left(1-P_S\right)}{1-P_S P_R \left(2-P_S\right)},
\end{eqnarray}
we propose the following link selection variable $d_i$ to be used:
\begin{eqnarray}\label{sol-d-delay-1}
&&\hspace{-6mm}\textit{If $Q(i-1)\leq R_0$ and $O_S(i)=1$, then $d_i=0$,}\nonumber\\
 &&\hspace{-6mm}\textit{otherwise $d_i$ is given by (\ref{sol-d-1}).}
\end{eqnarray} 
\textbf{Case 2:}
 If $P_R<1/(2-P_S)$ and the required delay $E\{T\}$ satisfies
\begin{eqnarray}\label{rr-2}
   \frac{1}{1-P_R \left(2-P_S\right)} < E\{T\}  &&\hspace{-6mm} \leq \frac{1}{1-P_R \left(2-P_S\right)} \nonumber\\
&&\hspace{-6mm}+\frac{2 \left(1-P_S\right)}{1-P_S P_R \left(2-P_S\right)},
\end{eqnarray}
we propose the following link selection variable $d_i$ to be used:
\begin{eqnarray}\label{sol-d-delay-2}
 &&\hspace{-6mm}\textit{If $Q(i-1)=0$ and $O_S(i)=1$, then $d_i=0$,} \nonumber\\ 
&&\hspace{-6mm}\textit{otherwise $d_i$ is given by (\ref{sol-d-1}).}
\end{eqnarray} 
\textbf{Case 3:}  
If the required delay $E\{T\}$ satisfies
\begin{eqnarray}\label{rr-3}
  \frac{1}{1-P_R } <  E\{T\} \leq  \frac{1}{1-P_R \left(2-P_S\right)} ,
\end{eqnarray}
we propose the following link selection variable $d_i$ to be used:
\begin{eqnarray}\label{sol-d-delay-3}
&&\hspace{-6mm}\textit{If $Q(i-1)=0$ and $O_S(i)=1$, then $d_i=0$,} \nonumber\\
&&\hspace{-6mm}\textit{otherwise $d_i$ is given by (\ref{sol-d-2}).}
\end{eqnarray} 
For each of the proposed link selection variables $d_i$, the required delay can be met by adjusting the value of $P_C={\rm Pr}\{\mathcal{C}=1\}$, where the minimum and maximum delays are achieved with $P_C=1$ and $P_C=0$, respectively. 
\end{proposition}
\begin{remark}\label{remark_limits} 
The delay limits given by (\ref{rr-1}), (\ref{rr-2}), and (\ref{rr-3}) arise from the analysis of the proposed protocols with link selection variables (\ref{sol-d-delay-1}), (\ref{sol-d-delay-2}), and (\ref{sol-d-delay-3}), respectively. We will investigate these
delay limits in Lemma~\ref{lemma5a} in Section \ref{s4C} and the corresponding proof is provided in Appendix \ref{proof-lemma5a}.
\end{remark}
\begin{remark}\label{remark_small_delay} 
We have not proposed a buffer-aided relaying protocol with adaptive link selection that can satisfy a required delay smaller than $1/(1-P_R)$. For such small delays, Conventional Relaying 2 without adaptive link selection can be used. 
\end{remark}
\subsection{Throughput and Delay\label{s4C}}
In the following, we analyze the throughput, the average delay, and the probability of having $k$ packets in the queue for the modified link selection protocols proposed in Proposition~\ref{proposition1} in the previous subsection. The results are summarized in the
following theorem.

\begin{theorem}\label{theorem5}
Consider a buffer-aided relaying system operating in temporally uncorrelated block fading. Let source and relay transmit with rate $R_0$, respectively, and let the buffer size at the relay be limited to $L$ packets each comprised of  $R_0$ bits. 
Assume that the relay drops newly received packets if the buffer is full. Then, depending on the adopted link selection protocol, the following cases can be distinguished:
\\
\textbf{Case 1:}
If the link selection variable $d_i$ is given by (\ref{sol-d-delay-1}), the  probability of the buffer having $k$ packets in its queue, ${\rm Pr}\{Q=k R_0 \}$, is obtained as
\begin{eqnarray}
 &&\hspace{-6mm}{\rm Pr}\{Q=kR_0\}=\nonumber\\
&&\hspace{-8.5mm}\left\{\hspace{-2mm}\begin{array}{ll}
 \frac{ p^{L-1}(2 p+q-1)(P_S-q)} {p^{L-1}(2p (1-q)+q(2-q)-P_S(2-P_S))-(1-p-q)^{L-1}(1-P_S)^2}, &\hspace{-3mm} k=0\\
 \frac{ p^{L-1}(2 p+q-1)(1-P_S)}{p^{L-1}(2p (1-q)+q(2-q)-P_S(2-P_S))-(1-p-q)^{L-1}(1-P_S)^2}, &\hspace{-3mm} k = 1\\
\frac{ p^{L-k}(2 p+q-1)(1-P_S)^2(1-p-q)^{k-2}}{p^{L-1}(2p (1-q)+q(2-q)-P_S(2-P_S))-(1-p-q)^{L-1}(1-P_S)^2}, &\hspace{-3mm} k\hspace{-0.5mm}=\hspace{-0.5mm}2..L
\end{array}\right.\nonumber
\end{eqnarray}
\vspace{-7mm}
\begin{eqnarray}\label{eq-Q-L-1}
\end{eqnarray}
where $p$ and $q$ are given by
\begin{eqnarray}\label{p&q-1}
p=(1-P_S)(1-P_R)P_C+P_S(1-P_R)\;;\quad q=P_SP_R \label{eq-q} .
\end{eqnarray}
Furthermore, the average queue length, $E\{Q\}$, the average delay, $E\{T\}$, and throughput, $\tau$, are given by
\begin{eqnarray}
 E\{Q\}=&&\hspace{-6mm} R_0\frac{1-P_S}{2 p+q-1} \nonumber\\
 \times &&\hspace{-6mm}
\big[ p^{L-1} \left((2 p+q)^2-p-q-P_S (3 p+q-1)\right)\nonumber\\
&&\hspace{-6mm}-(1-P_S) (1-p-q)^{L-1} (L (2 p+q-1)+p)\big]\nonumber\\
/
&&\hspace{-6mm}\big[p^{L-1} (2 p (1-q)+(2-q) q-(2-P_S) P_S)\nonumber\\
&&\hspace{-6mm}-(1-P_S)^2 (1-p-q)^{L-1}\big]
\label{mean-Q-1}
\end{eqnarray}
\begin{eqnarray}
 E\{T\}= &&\hspace{-6mm}\frac{1}{2 p+q-1} \nonumber\\
 \times &&\hspace{-6mm}
\big[p^{L-1} \left((2 p+q)^2-P_S (3 p+q-1)-p-q\right)\nonumber\\
&&\hspace{-6mm}-(1-P_S) (1-p-q)^{L-1} (L (2 p+q-1)+p)\big]\nonumber\\
/&&\hspace{-6mm}
\big[p^{L-1} (P_S (p+q-1)-q (2 p+q)+p+q)\nonumber\\
&&\hspace{-6mm}-(1-P_S) p (1-p-q)^{L-1}\big]
\label{delay-1}
\end{eqnarray}
\begin{eqnarray}
 \tau=&&\hspace{-6mm}(1-P_S)
\big[(1-P_S) p (1-p-q)^{L-1}\nonumber\\
&&\hspace{-6mm}+p^{L-1} (P_S (1-p-q)+q (2 p+q)-p-q)\big]\nonumber\\
/&&\hspace{-6mm}
\big[p^{L-1} ((2-P_S) P_S-2 p (1-q)-(2-q) q)\nonumber\\
&&\hspace{-6mm}+(1-P_S)^2 (1-p-q)^{L-1}\big].\label{trup-delay-1}
\end{eqnarray}
\textbf{Case 2:}
If link selection variable $d_i$ is given by either (\ref{sol-d-delay-2}) or (\ref{sol-d-delay-3}), the  probability of the buffer having $k$ packets in its queue, ${\rm Pr}\{Q=k R_0 \}$, is given by
\begin{eqnarray}
 {\rm Pr}\{Q=kR_0\}=\hspace{-0.5mm}\left\{\hspace{-2.5mm}
\begin{array}{ll}
 \frac{p^L (2 p+q-1)}{p^L (2 p+q-P_S)-(1-P_S) (1-p-q)^L}, & \hspace{-3mm} k=0\\
\frac{(1-P_S) (2 p+q-1) p^{L-k} (1-p-q)^{k-1}}{p^L (2 p+q-P_S)-(1-P_S) (1-p-q)^L}, & \hspace{-3mm} k=1... L
\end{array}\right.\nonumber
\end{eqnarray}
\vspace{-7mm}
\begin{eqnarray}\label{eq-Q-L-2}
\end{eqnarray}
where, if link selection variable $d_i$ is given by (\ref{sol-d-delay-2}), $p$ and $q$ are given by (\ref{p&q-1}), while if link selection variable $d_i$ is given by (\ref{sol-d-delay-3}), $p$ and $q$ are given by 
\begin{eqnarray}\label{p&q-2}
    p=1-P_R \quad {\rm and} \quad q=P_S P_R +(1-P_S) P_R P_C.
\end{eqnarray}
Furthermore, the average queue length, $E\{Q\}$, the average delay, $E\{T\}$, and  throughput, $\tau$,  are given by
\begin{eqnarray}
       E\{Q\}=&&\hspace{-6mm}R_0\frac{1-P_S}{2 p+q-1}\nonumber\\
\times&&\hspace{-6mm}\frac{ p^{L+1}-(1-p-q)^L (L (2 p+q-1)+p)}{ p^L (2 p+q-P_S)-(1-P_S) (1-p-q)^L},\quad
\label{mean-Q-2}
\end{eqnarray}
\begin{eqnarray}
E\{T\}=&&\hspace{-6mm}\frac{1}{2p+q-1}\frac{1}{p}\nonumber\\
\times&&\hspace{-6mm}
\frac{p^{L+1}-(1-p-q)^L (L (2 p+q-1)+p)}{p^L-(1-p-q)^L},\quad
\label{delay-2}
\end{eqnarray}
\begin{equation}
\tau=R_0(1-P_S) p \frac{p^L-(1-p-q)^L}{p^L (2 p+q-P_S)-(1-P_S) (1-p-q)^L}.\label{trup-delay-2}
\end{equation}
\end{theorem}
\begin{proof}
Please refer to Appendix \ref{proof-th-5}. 
\end{proof}
Due to their complexity, the equations in Theorem \ref{theorem5} do not provide much insight into the performance of the considered system. To overcome this problem, we consider the case 
$L\gg 1$, which leads to significant simplifications and design insight. This is addressed in the following lemma.
\begin{lemma}\label{lemma5}
For the system considered in Theorem \ref{theorem5}, assume that $L\to \infty$. In this case, for a system with link selection variable $d_i$ given by  (\ref{sol-d-delay-1}), (\ref{sol-d-delay-2}), or (\ref{sol-d-delay-3}) to 
be able to achieve a fixed delay, $E\{T\}$, that does not grow with $L$ as $L\to \infty$,  the condition $2p+q-1>0$ must hold. If $2p+q-1>0$ holds, the following simplifications can be made for each of the considered link selection variables:
\\
\textbf{Case 1:}
If  the link selection variable $d_i$ is given by (\ref{sol-d-delay-1}), the probability of the buffer being empty, the average delay, $E\{T\}$, and throughput, $\tau$, simplify to 
\begin{equation}
 {\rm Pr}\{Q=0\} =P_S \frac{ 2 P_C (1-P_R) (1-P_S)+(2-P_R) P_S-1} 
{2 P_C (1-P_S)(1-P_S P_R)+P_S^2 (1-P_R)} \label{eq-Q_0-app-1}
\end{equation} 
\begin{eqnarray}
&&\hspace{-6mm} E\{T\}=\frac{1}{2 P_C (1-P_R) (1-P_S)-P_R P_S+2 P_S-1}\qquad\nonumber\\
 &&\hspace{-2mm}+\frac{2 P_C (1-P_S)}{P_S^2 (P_C (2 P_R-1)-P_R+1)-2 P_C P_R P_S+P_C} 
\label{delay-app-1}
\end{eqnarray} 
\begin{eqnarray}
\tau=&&\hspace{-6mm}R_0 (1-P_S)\nonumber\\
\times&&\hspace{-6mm}\frac{P_S^2 (P_C (2 P_R-1)-P_R+1)-2 P_C P_R P_S+P_C}{2 P_C (1-P_S) (1-P_S
   P_R)+(1-P_R) P_S^2} .\qquad\label{trup-delay-app-1}
\end{eqnarray} 
\textbf{Case 2:}
If  the link selection variable $d_i$ is given by  (\ref{sol-d-delay-2}),  the probability of the buffer being empty, the average delay, $E\{T\}$, and  throughput, $\tau$, simplify to 
\begin{eqnarray}
    {\rm Pr}\{Q\hspace{-0.6mm}=\hspace{-0.6mm}0\} \hspace{-0.6mm}=&&\hspace{-6mm}\frac{2 P_C (1-P_R) (1-P_S)+P_S(2-P_R)-1}{(1-P_R) (P_S+2 P_C (1-P_S))}\qquad\label{eq-Q_0-app-2}\\
  E\{T\}=&&\hspace{-6mm}\frac{1}{2 P_C (1-P_R) (1-P_S)-P_R P_S+2 P_S-1}
\label{delay-app-2}\;\\
  \tau=&&\hspace{-6mm}R_0 (1-P_S)\frac{ P_C (1-P_S)+P_S}{2 P_C (1-P_S)+P_S}. \label{trup-delay-app-2}
\end{eqnarray} 
\textbf{Case 3:}
If the link selection variable $d_i$ is given by (\ref{sol-d-delay-3}), the probability of the buffer being empty, the average delay, $E\{T\}$, and the throughput, $\tau$, simplify to 
\begin{eqnarray}
    {\rm Pr}\{Q=0\} =&&\hspace{-6mm}\frac{1-P_R  (2-P_S-P_C(1- P_S))}{2-P_S-P_R (2-P_S-P_C (1-P_S))}
    \label{eq-Q_0-app-3}\qquad\quad\\
  E\{T\}=&&\hspace{-6mm}\frac{1}{1-P_R  (2-P_S-P_C(1- P_S))}
\label{delay-app-3}\\
  \tau=&&\hspace{-6mm}R_0 \frac{1+P_S P_R -P_R-P_S}{2-P_S-P_R (2-P_S-P_C (1-P_S))}. \label{trup-delay-app-3}
\end{eqnarray} 
For each of the considered cases, the probability $P_C$ can be used to adjust the desired average delay $E\{T\}$ in (\ref{delay-app-1}), (\ref{delay-app-2}), and (\ref{delay-app-3}).
\end{lemma}
\begin{proof}
Please refer to Appendix \ref{5a}.
\end{proof}
As already mentioned in Proposition~\ref{proposition1}, it is not possible to achieve any desired average delay with the proposed buffer-aided link selection protocols. The limits of the achievable average delay for each of the proposed  link selection 
variables $d_i$ in Proposition~\ref{proposition1} are provided in the following lemma.
\begin{lemma}\label{lemma5a}
Depending on the adopted link selection variable $d_i$ the following cases can be distinguished for the average delay:\\
\textbf{Case 1:} If the link selection variable $d_i$ is given by  (\ref{sol-d-delay-1}), then if  $P_R<1/(2-P_S)$ and $P_S<1/(2-P_R)$, the system can achieve any average delay $E\{T\}\geq T_{\rm min,1}$, where $T_{\rm min,1}$ is given by
\begin{eqnarray}\label{delay-app-min-1}
   T_{\rm min,1}= \frac{1}{1-P_R \left(2-P_S\right)} +\frac{2 \left(1-P_S\right)}{1-P_S P_R \left(2-P_S\right)}.
\end{eqnarray}
On the other hand, if  $P_R<1/(2-P_S)$ and $P_S>1/(2-P_R)$, the system can  achieve any average delay in the interval $T_{\rm min,1}\leq E\{T\}\leq T_{\rm max,1}$, where $T_{\rm max,1}$
is given by 
\begin{eqnarray}\label{delay-app-max-1}
    T_{\rm max,1}=  \frac{1}{P_S(2-P_R)-1}.
\end{eqnarray} 
\textbf{Case 2:} If  the link selection variable $d_i$ is given by  (\ref{sol-d-delay-2}), then if  $P_R<1/(2-P_S)$ and $P_S<1/(2-P_R)$, the system can achieve any average delay $E\{T\}\geq T_{\rm min,2}$, where  $T_{\rm min,2}$ is given by
\begin{eqnarray}\label{delay-app-min-2}
    T_{\rm min,2}=
\frac{1}{1-P_R(2-P_S)}.
\end{eqnarray} 
However, if $P_R<1/(2-P_S)$ and $P_S>1/(2-P_R)$, the system can achieve any average delay $T_{\rm min,2}\leq E\{T\}\leq T_{\rm max,2}$, where $T_{\rm max,2}=T_{\rm max,1}$.
\\
\textbf{Case 3:} If  the link selection variable $d_i$ is given by (\ref{sol-d-delay-3}), then if  $P_R>1/(2-P_S)$, the system can achieve any average delay $E\{T\}\geq T_{\rm min,3}$, where  $T_{\rm min,3}$ is given by
\begin{eqnarray}\label{delay-app-min-3}
    T_{\rm min,3}=  
\frac{1}{1-P_R}.
\end{eqnarray} 
On the other hand, if  $P_R<1/(2-P_S)$, the system can achieve any average delay $T_{\rm min,3}\leq E\{T\}\leq T_{\rm max,3}$, where $T_{\rm max,3}=T_{\rm min,2}$.
\end{lemma}
\begin{proof}
Please refer to Appendix \ref{proof-lemma5a}.
\end{proof}
In the following, we investigate the outage probability of the proposed buffer-aided relaying protocol for delay constrained fixed rate transmission.
\subsection{Outage Probability\label{s4D}}
The following theorem specifies the outage probability.
\begin{theorem}\label{theorem6}
For the considered buffer-aided relaying protocol in Proposition~\ref{proposition1}, if the required delay can be satisfied by using the link selection variable $d_i$ in either (\ref{sol-d-delay-1}) or (\ref{sol-d-delay-2}), the outage probability is given by 
\begin{eqnarray}\label{OP-delay-1-and-2}
    F_{\rm out} =&&\hspace{-6mm} P_S {\rm Pr}\{Q=0  \}\nonumber\\
+&&\hspace{-6.5mm}  P_S P_R\big(1-{\rm Pr}\{Q=0  \}-  {\rm Pr}\{Q=LR_0  \}\big) \nonumber\\
+&&\hspace{-6.5mm} 
\big((1\hspace{-0.5mm} -\hspace{-0.5mm} P_S) P_R +(1\hspace{-0.5mm} -\hspace{-0.5mm} P_S P_R)(1\hspace{-0.5mm} - \hspace{-0.5mm} P_C) \big) {\rm Pr}\{Q=LR_0  \} ,\nonumber\\
\end{eqnarray}
where if $d_i$ is given by  (\ref{sol-d-delay-1}), ${\rm Pr}\{Q=0\}$  and ${\rm Pr}\{Q=LR_0\}$ are given by (\ref{eq-Q-L-1}) with $p$ and $q$  given by (\ref{p&q-1}). On the other hand, if $d_i$ is given by  (\ref{sol-d-delay-2}), ${\rm Pr}\{Q=0\}$  and 
${\rm Pr}\{Q=LR_0\}$ are given by (\ref{eq-Q-L-2}) with $p$ and $q$ given by (\ref{p&q-1}).

If the required delay is satisfied by using the link selection variable  $d_i$ given by  (\ref{sol-d-delay-3}), then the outage probability is given by 
\begin{eqnarray}\label{OP-delay-3}
    F_{\rm out} =&&\hspace{-6mm} P_S {\rm Pr}\{Q=0  \}\nonumber\\
+&&\hspace{-6mm}  P_S P_R\big(1-{\rm Pr}\{Q=0  \}- {\rm Pr}\{Q=LR_0  \}\big) \nonumber\\
+&&\hspace{-6mm} (1-P_S) P_R (1- P_C)  {\rm Pr}\{Q=LR_0  \} ,
\end{eqnarray}
where ${\rm Pr}\{Q=0\}$  and ${\rm Pr}\{Q=LR_0\}$ are given by (\ref{sol-d-delay-2}) with $p$ and $q$  given by (\ref{p&q-2}).
\end{theorem}
\begin{proof}
Please refer to Appendix \ref{app_J}.
\end{proof}
The expressions for $F_{\rm out}$ in Theorem \ref{theorem6} are valid for general $L$. However, significant  simplifications are possible if $L\gg 1$.  This is addressed in the following lemma.
\begin{lemma}
When $L\to \infty$, the outage probability  given by (\ref{OP-delay-1-and-2}) and (\ref{OP-delay-3}) simplifies to
\begin{eqnarray}\label{OP-delay-1-and-2_a}
    F_{\rm out} = P_S {\rm Pr}\{Q=0  \}+  P_S P_R\big(1-{\rm Pr}\{Q=0  \}\big),
\end{eqnarray}
where ${\rm Pr}\{Q=0\}$  is given by (\ref{eq-Q_0-app-1}), (\ref{eq-Q_0-app-2}), and (\ref{eq-Q_0-app-3}) if $d_i$ is given by (\ref{sol-d-delay-1}),  (\ref{sol-d-delay-2}),  and (\ref{sol-d-delay-3}), respectively.
\end{lemma}
\begin{proof}
Eq.~(\ref{OP-delay-1-and-2_a}) is obtained by letting ${\rm Pr}\{Q=LR_0  \}\to 0$ when $L\to\infty$ in (\ref{OP-delay-1-and-2}) and (\ref{OP-delay-3}).
\end{proof}
The expression for the outage probability in (\ref{OP-delay-1-and-2_a}) can be further simplified in the high SNR regime, which provides insight into the achievable diversity gain. This is summarized in the following theorem.
\begin{theorem}\label{theo-outage-high-snr}
 In the high SNR regime, when $\gamma_S=\gamma_R=\gamma\to \infty$, depending on the required delay that the system has to satisfy, two cases can be distinguished:  
\\
\textbf{Case 1:} If   $1<E\{T\}\leq 3$, the outage probability asymptotically converges to  
\begin{eqnarray}\label{OP-h-snr-2}
    F_{\rm out} \to \frac{P_S} {E\{T\}+1}, \quad \textrm{as }\gamma\to \infty.
\end{eqnarray}
\textbf{Case 2:} If   $E\{T\}>3$, the outage probability  asymptotically converges to  
\begin{eqnarray}\label{OP-h-snr-1}
    F_{\rm out} \to\frac{P_S^2} {E\{T\}-1}+P_S P_R, \quad \textrm{as }\gamma\to \infty .
\end{eqnarray}
Therefore, assuming Rayleigh fading, the considered system  achieves a diversity gain of two if and only if  $E\{T\}>3$.
\end{theorem}
\begin{proof}
Please refer to Appendix \ref{proof-outage-high-snr}.
\end{proof}
According to Theorem~\ref{theo-outage-high-snr}, for Rayleigh fading, a diversity gain of two can be also achieved for delay constrained transmission, which underlines the appeal of buffer-aided relaying with adaptive link selection compared
to conventional relaying, which only achieves a diversity gain of one even in case of infinite delay (Conventional Relaying 1).
\section{Mixed Rate Transmission}\label{sec-hybrid}
In  this section, we investigate buffer-aided relaying protocols with adaptive link selection for mixed rate transmission. In particular, we assume that the source does not have CSIT and transmits with fixed rate $S_0$ but the relay has full CSIT 
and transmits with the maximum possible rate, $R(i)=\log_2(1+r(i))$, that does not cause an outage in the $\mathcal{R}$-$\mathcal{D}$ channel. For this scenario, we consider first delay unconstrained transmission and derive the optimal
link adaptive buffer-aided relaying protocols with and without power allocation. Subsequently, we investigate the impact of delay constraints. 

Before we proceed, we note that for mixed rate transmission the throughput can be expressed as
\begin{eqnarray}\label{eq_tau}
    \tau=\lim_{N\to\infty}\frac{1}{N}\sum_{i=1}^N d_i   \min\{\log_2(1+r(i)),Q(i-1)\},
\end{eqnarray}
where we used (\ref{3-lit}) and (\ref{eq1}).
For the derivation of the maximum throughput of buffer-aided relaying with adaptive link selection the following theorem is useful.
\begin{theorem}\label{theorem1-mixed}
The link selection policy that maximizes the throughput of the considered buffer-aided relaying system for mixed rate transmission can  be found in the set of link selection policies that satisfy   
\begin{equation}\label{cond-1-mixed}
    \lim_{N\to\infty}\frac{1}{N}\sum_{i=1}^N (1-d_i) O_S(i)  S_0 = \lim_{N\to\infty}\frac{1}{N}\sum_{i=1}^N   d_i \log_2(1+r(i))\;.
\end{equation}
Furthermore, for link selection policies within this set, the throughput is given by the right (and left) hand side of (\ref{cond-1-mixed}).
\end{theorem}
\begin{proof}
A proof of this theorem can obtained by  replacing $O_R(i) R_0$ by $\log_2(1+r(i))$ in the proof of Theorem~\ref{theorem1} given in Appendix~\ref{app_for_t_1}. 
\end{proof}
Hence, similar to fixed rate transmission, for the set of policies considered in Theorem \ref{theorem1-mixed}, for $N\to\infty$, the buffer at the relay is practically always fully backlogged. Thus, the 
$\min(\cdot)$ function in (\ref{eq_tau}) can be omitted and the throughput is given by the right hand side of (\ref{cond-1-mixed}). 
\subsection{Optimal Link Selection Policy Without Power Allocation}\label{sec_mixed_no_delay}
Since the relay has the instantaneous CSI of both links, it can also optimize its transmit power. However, to get more insight, we first consider the case where the relay transmits with fixed power. We note that power allocation is not always desirable as it
requires highly linear power amplifiers and thus, increases the implementation complexity of the relay. 

According to Theorem~\ref{theorem1-mixed}, the optimal link selection policy maximizing the throughput can be found in the set of policies that satisfy (\ref{cond-1-mixed}). 
Therefore, the optimal policy can be obtained from the following optimization problem
\begin{eqnarray}\label{MPR-mixed-1}
\begin{array}{ll}
 {\underset{d_i}{\rm{Maximize: }}}&\frac{1}{N}\sum_{i=1}^N d_i \log_2\big(1+r(i)\big) \\
{\rm{Subject\;\; to: }} &{\rm C1:}\, \frac{1}{N}\sum_{i=1}^N (1-d_i) O_S(i) S_0\\
&\qquad  =\frac{1}{N}\sum_{i=1}^N d_i \log_2\big(1+r(i)\big) \\
  &{\rm C2:} \, d_i(1-d_i)=0,\; \forall i ,\\
\end{array}
\end{eqnarray}
where $N\to\infty$, constraint C1 ensures that the search for the optimal policy is conducted only among the policies that satisfy (\ref{cond-1-mixed}), and C2 ensures that $d_i\in\{0,1\}$. The solution of (\ref{MPR-mixed-1}) leads to the following theorem.
 
\begin{theorem}\label{theorem7} Let   the  pdfs of $s(i)$ and $r(i)$ be denoted by $f_s(s)$ and $f_r(r)$, respectively. Then, for the considered  buffer-aided relaying system in which the source transmits with a fixed rate $S_0$ and fixed power $\mathcal{P_S}$, 
and the relay  transmits with an adaptive rate $R(i)=\log_2(1+r(i))$ and fixed power $\mathcal{P_R}$, two cases have to be distinguished for the optimal link selection variable $d_i$, which maximizes the throughput:
\\
\textbf{Case 1: } If 
\begin{eqnarray}\label{cond-PS-mixed-non-pa-1}
    P_S\leq \frac{S_0}{S_0+\int_0^\infty \log_2(1+r) f_r(r) dr}
\end{eqnarray}
holds, then
\begin{eqnarray}\label{sol-d-mixed-1}
d_i=\left\{
\begin{array}{cl}
1 & \textrm{if }   O_S(i)=0\\
1 & \textrm{if }   O_S(i)=1 \textrm{ AND } r(i)\geq 2^{\rho S_0}-1\\
0 & \textrm{if }   O_S(i)=1 \textrm{ AND } r(i)<2^{\rho S_0}-1 \;,\\
\end{array} 
\right.
\end{eqnarray} 
where $\rho$ is a constant which can be found as the solution of 
\begin{eqnarray}\label{eq11}
 S_0(1-P_S) \int_{0}^{2^{\rho S_0}-1}\hspace{-8mm}f_r(r) dr =&&\hspace{-6mm}P_S\int_{0}^{\infty} \hspace{-2mm}\log_2(1+r) f_r(r) dr\nonumber\\ 
&&\hspace{-22mm}+(1-P_S)\int_{2^{\rho S_0}-1}^{\infty}\hspace{-6mm}\log_2(1+r)f_r(r) dr \;.
\end{eqnarray}
In this case, the maximum throughput is given by the right (and left) hand side of (\ref{eq11}).
\\ 
\textbf{Case 2: } If  (\ref{cond-PS-mixed-non-pa-1}) does not hold, then
\begin{eqnarray}\label{sol-d-mixed-2}
d_i=\left\{
\begin{array}{cl}
0 & \textrm{if }   O_S(i)=1\\
1 & \textrm{if }   O_S(i)=0 \;.\\
\end{array} 
\right.
\end{eqnarray} 
In this case, the maximum throughput is  given by
\begin{eqnarray}\label{max-tau-mixed-non-pa-2}
    \tau=S_0(1-P_S)\;.
\end{eqnarray}
\end{theorem}

\begin{proof}
Please refer to Appendix \ref{app_K}.
\end{proof}

We note that with mixed rate transmission the ${\cal S}$-${\cal R}$ link is used only if it is not in outage, cf.~(\ref{sol-d-mixed-1}), (\ref{sol-d-mixed-2}). On the other hand, the ${\cal R}$-${\cal D}$ link is never in outage since the transmission rate is 
adjusted to the channel conditions. Furthermore, buffer-aided relaying with adaptive link selection has a larger throughput than Conventional Relaying 1, and also achieves a multiplexing gain of one.


To get more insight, we specialize the results derived thus far in this section to Rayleigh fading links.
\begin{lemma}\label{lemma7}
For Rayleigh fading links, condition (\ref{cond-PS-mixed-non-pa-1}) simplifies to
\begin{equation}\label{cond-PS-mixed-non-pa-1-ray}
      P_S= 1-\exp\left(-\frac{2^{ S_0}-1}{\Omega_S}\right) \leq \frac{S_0}{S_0+e^{1/\Omega_R} E_1(1/\Omega_R)/\ln(2)}.
\end{equation}
Furthermore, (\ref{eq11})  simplifies to
\begin{eqnarray}\label{eq12}
&&\hspace{-6mm} S_0 \exp\left(-\frac{2^{ S_0}-1}{\Omega_S}\right) \left[1-\exp\left(-\frac{2^{\rho S_0}-1}{\Omega_R}\right) \right]
\nonumber\\
&&\hspace{-6mm} = \frac{ e^{1/\Omega_R} }{\ln(2)}\Bigg[\left(1-\exp\left(-\frac{2^{ S_0}-1}{\Omega_S}\right)\right) E_1\left(\frac{1}{\Omega_R}\right) \nonumber\\
&&\hspace{-5mm}+\exp\left(-\frac{2^{ S_0}-1}{\Omega_S}\right) E_1\left(\frac{2^{\rho S_0}}{\Omega_R}\right) \Bigg] 
\nonumber\\
&&\hspace{-6mm}+ \exp\left(-\frac{2^{\rho S_0}-1}{\Omega_R}\right) \exp\left(-\frac{2^{ S_0}-1}{\Omega_S}\right) \rho S_0\;,
\end{eqnarray}
and the maximum throughput is given by the right (and left) hand side of (\ref{eq12}). If (\ref{cond-PS-mixed-non-pa-1-ray}) does not hold, the throughput can be obtained by simplifying
(\ref{max-tau-mixed-non-pa-2}) to
\begin{eqnarray}\label{max-tau-mixed-non-pa-2-ray}
    \tau=S_0 \exp\left(-\frac{2^{ S_0}-1}{\Omega_S}\right)\;.
\end{eqnarray}
\end{lemma}
\begin{proof}
Equations (\ref{cond-PS-mixed-non-pa-1-ray})-(\ref{max-tau-mixed-non-pa-2-ray}) are  obtained  by inserting the pdfs of $s(i)$ and $r(i)$ into (\ref{cond-PS-mixed-non-pa-1}), (\ref{eq11}), and (\ref{max-tau-mixed-non-pa-2}),
respectively.
\end{proof}

\subsection{Optimal Link Selection Policy with Power Allocation}\label{sec_mixed_no_delay_pa}
As mentioned before, since for mixed rate transmission the relay is assumed to have the full CSI of both links, power allocation can be applied to further improve performance. In other words, the relay can adjust its transmit power $\mathcal{P_R}(i)$ to the
channel conditions while the source still transmits with fixed power $\mathcal{P_S}(i)=\mathcal{P_S}$, $\forall i$. In the following, for convenience, we will use the transmit  SNRs without fading, $\gamma_S$ and $\gamma_R(i)$, which may be 
viewed as normalized powers, as variables instead of the actual powers $\mathcal{P_S}=\gamma_S \sigma_{n_R}^2$ and $\mathcal{P_R}(i)=\gamma_R(i)\sigma_{n_D}^2$. 

For the power allocation case, Theorem \ref{theorem1-mixed} is still applicable but it is convenient to rewrite the throughput as
\begin{eqnarray}
    \tau=\lim_{N\to \infty}\frac{1}{N}\sum_{i=1}^N d_i \log_2(1+ \gamma_R(i) h_R(i)).
\end{eqnarray}
We note that (\ref{cond-1-mixed}) also applies to the case of power allocation.
Furthermore, in order to meet the average power constraint $\Gamma$, the instantaneous (normalized) power $\gamma_R(i)$ and the fixed (normalized)  power $\gamma_S$ have to satisfy the following condition:
\begin{equation}\label{cond-mixed-2-PA}
  \lim_{N\to \infty}\frac{1}{N}\sum_{i=1}^N (1-d_i) O_S(i) \gamma_S  +\lim_{N\to \infty}\frac{1}{N}\sum_{i=1}^N d_i \gamma_R(i)\leq\Gamma.
\end{equation}
Thus, the optimal link selection policy for mixed rate transmission is the solution of the following optimization problem:
\begin{eqnarray}
\begin{array}{ll}
 {\underset{d_i,\gamma_R(i)}{\rm{Maximize: }}}&\frac{1}{N}\sum_{i=1}^N d_i \log_2\big(1+\gamma_R(i) h_R(i)\big)  \\
{\rm{Subject\;\; to: }} &{\rm C1:}\, \frac{1}{N}\sum_{i=1}^N (1-d_i) O_S(i) S_0\\
&\qquad  =\frac{1}{N}\sum_{i=1}^N d_i \log_2\big(1+\gamma_R(i) h_R(i)\big) \\
  &{\rm C2:} \, d_i(1-d_i)=0\;,\quad \forall i\\
 &{\rm C3:} \, \frac{1}{N}\sum_{i=1}^N  (1-d_i) O_S(i)\gamma_S  \nonumber\\
&\qquad+\frac{1}{N}\sum_{i=1}^N d_i \gamma_R(i)\leq\Gamma , \\
\end{array} 
\nonumber
\end{eqnarray}
\vspace{-10mm}
\begin{eqnarray}
    \label{MPR-mixed-1-PA}
\end{eqnarray}
where $N\to\infty$, constraints C1 and C3 ensure that the search for the optimal policy is conducted only among those policies that jointly satisfy (\ref{cond-1-mixed}) and the  source-relay power constraint (\ref{cond-mixed-2-PA}), respectively,  and C2 
ensures that $d_i\in\{0,1\}$. The solution of (\ref{MPR-mixed-1-PA}) is provided in the following theorem.

\begin{theorem}\label{theorem8}
Let the pdfs of $h_S(i)$ and $h_R(i)$ be denoted by $f_{h_S}(h_S)$ and $f_{h_R}(h_R)$, respectively. Then, for the considered buffer-aided relaying system where the source transmits with a fixed rate $S_0$ and fixed power $\gamma_S$ and 
the relay  transmits with adaptive rate $R(i)=\log_2(1+r(i))=\log_2(1+\gamma_R(i) h_R(i))$ and adaptive power $\gamma_R(i)$, two cases have to be considered for the optimal link selection variable $d_i$ which maximizes the throughput:
\\
\textbf{Case 1: } If 
\begin{eqnarray}\label{cond-PS-mixed-pa-1}
    P_S\leq \frac{S_0}{S_0+\int_{\lambda_t}^\infty \log_2(h_R/\lambda_t) f_{h_R}(h_R) d h_R},
\end{eqnarray}
holds, where ${\lambda_t}$ is found as the solution to
\begin{eqnarray}\label{lambda-cond-PS-mixed-pa-1}
      P_S \int_{\lambda_t}^\infty \left(\frac{1}{\lambda_t}-\frac{1}{h_R}\right) f_{h_R}(h_R) d h_R 
+ \gamma_S (1-P_S)  =\Gamma,
\end{eqnarray}
then the optimal  power $\gamma_R(i)$ and link selection variable $d_i$ which maximize the throughput are given  by  
\begin{eqnarray}
\gamma_R(i)&=&\max\left\{0,\frac{1}{\lambda}-\frac{1}{h_R(i)}\right\}\; ,  
\label{power-eq-2b}
\end{eqnarray} 
and
\begin{eqnarray}
d_i=\left\{
\begin{array}{cl}
1 & \textrm{if }  O_S(i)=0 \textrm{ AND } h_R(i)\geq\lambda \\
1 & \textrm{if }   O_S(i)=1 \textrm{ AND }   h_R(i)\geq\lambda \\
&\quad  \textrm{ AND } \ln\left(\frac{h_R(i)}{\lambda}\right)+\frac{\lambda}{h_R(i)}\geq\rho S_0 - \lambda \gamma_S+1 \\
0 & \textrm{if }   O_S(i)=1 \textrm{ AND } h_R(i)<\lambda\\
0 & \textrm{if }   O_S(i)=1 \textrm{ AND }    h_R(i)\geq\lambda \\
&\quad  \textrm{ AND } \ln\left(\frac{h_R(i)}{\lambda}\right)+\frac{\lambda}{h_R(i)}<\rho S_0 - \lambda \gamma_S+1\\
\varepsilon & \textrm{if } O_S(i)=0 \textrm{ AND } h_R(i)<\lambda\; ,
\end{array} 
\right.\nonumber
\end{eqnarray} 
\vspace{-9mm}
\begin{eqnarray}\label{sol-d-mixed-PA}
\end{eqnarray}
where $\varepsilon$ is either $0$ or $1$ and has not impact on the throughput. Constants $\rho$ and $\lambda$ are  chosen such that constraints C1 and C3 in (\ref{MPR-mixed-1-PA}) 
are satisfied with equality.  These two constants can be found as the solution to the following system of equations
\begin{eqnarray}\label{nz-1}
&&\hspace{-6mm}S_0(1\hspace{-1mm} -\hspace{-1mm} P_S) \int_0^G\hspace{-3mm} f_{h_R}(h_R) d h_R=  P_S \hspace{-1mm} \int_\lambda^\infty \hspace{-3mm}\log_2\left(\frac{h_R}{\lambda}\right) f_{h_R}(h_R) d h_R \nonumber\\
&&\hspace{-6mm}+
(1-P_S) \int_G^\infty \log_2\left(\frac{h_R}{\lambda}\right) f_{h_R}(h_R) d h_R ,
\end{eqnarray} 
\begin{eqnarray}\label{nz-1a}
&& P_S \int_\lambda^\infty \left(\frac{1}{\lambda}-\frac{1}{h_R}\right) f_{h_R}(h_R) d h_R \nonumber\\
&& +
(1-P_S) \int_G^\infty \left(\frac{1}{\lambda}-\frac{1}{h_R}\right) f_{h_R}(h_R) d h_R\nonumber\\
&&+ \gamma_S (1-P_S) \int_0^G f_{h_R}(h_R) d h_R =\Gamma ,
\end{eqnarray}
where the integral limit $G$ is given by 
\begin{eqnarray}\label{int-lim-1}
    G=-\frac{\lambda }{W\{-e^{\lambda \gamma_S-\rho S_0-1}\}}\;.
\end{eqnarray}
Here, $W\{\cdot\}$ denotes the Lambert~$W$-function defined in \cite{corless1996lambertw}, which is available as built-in function in software packages such as Mathematica.
In this case, the maximized throughput is given by the right (and left) hand side of (\ref{nz-1}).

\noindent
\textbf{Case 2: } If (\ref{cond-PS-mixed-pa-1}) does not hold, the optimal power $\gamma_R(i)$ and link selection variable $d_i$ are given  by
\begin{eqnarray}
\gamma_R(i)=\max\left\{0,\frac{1}{\lambda}-\frac{1}{h_R(i)}\right\}\; ,  \textrm{ if } O_S(i)=0;
\label{power-eq-2b-2}
\end{eqnarray} 
\begin{eqnarray}\label{sol-d-mixed-PA-2}
d_i=\left\{
\begin{array}{cl}
0 & \textrm{if }  O_S(i)=1 \\
1 & \textrm{if }  O_S(i)=0 ,
\end{array} 
\right.
\end{eqnarray} 
where  $\lambda={\lambda_t}$ is the solution to (\ref{lambda-cond-PS-mixed-pa-1}). In this case, the maximum throughput is given by
\begin{eqnarray}\label{max-tau-mixed-pa-2}
    \tau=S_0(1-P_S).
\end{eqnarray}
\end{theorem}

\begin{proof}
Please refer to Appendix \ref{app_L}.
\end{proof}
\begin{remark}\label{remark_8}
Note that when conditions (\ref{cond-PS-mixed-non-pa-1}) and (\ref{cond-PS-mixed-pa-1}) do not hold, the throughput with and without power allocation is identical, cf.~(\ref{max-tau-mixed-non-pa-2}) and (\ref{max-tau-mixed-pa-2}). If conditions (\ref{cond-PS-mixed-non-pa-1}) 
and (\ref{cond-PS-mixed-pa-1}) do not hold, this means that the SNR in the ${\cal S}$-${\cal R}$ channel is low, whereas the SNR in the ${\cal R}$-${\cal D}$ channel is high. In this case, power allocation is not beneficial since the ${\cal S}$-${\cal R}$ channel is the bottleneck link, which
cannot be improved by power allocation at the relay. Furthermore, the throughput in (\ref{max-tau-mixed-non-pa-2}) and (\ref{max-tau-mixed-pa-2}) is identical to the throughput of a  point-to-point communication between the source and the relay since the number of time slots
required to transmit the information from the relay to the destination becomes negligible. Therefore, in this case, as far as the achievable throughput is concerned, the three-point half-duplex relay channel is transformed into a one hop channel
between source and relay. 
\end{remark}


In the following lemma, we concentrate on Rayleigh fading for illustration purpose. 

\begin{lemma}\label{lemma8}
For Rayleigh fading channels,  $P_S$ is given by
$$
P_S=1-\exp\left(-\frac{2^{ S_0}-1}{\gamma_S\bar\Omega_S}\right).
$$
Furthermore, condition (\ref{cond-PS-mixed-pa-1}) simplifies to
\begin{eqnarray}\label{cond-PS-mixed-pa-1-ray}
   P_S\leq \frac{S_0}{S_0+E_1(\lambda_t/\bar\Omega_R)/\ln(2)},
\end{eqnarray}
where $\lambda_t$ is found as the solution to
\begin{equation}\label{eq_un-1}
P_S\left[\frac{e^{-\lambda_t/\bar\Omega_R}}{\lambda_t} -\frac{1}{\bar\Omega_R} E_1\left(\frac{\lambda_t}{\bar\Omega_R}\right)\right] =\Gamma.
\end{equation}
For the case  where (\ref{cond-PS-mixed-pa-1-ray}) holds, (\ref{nz-1}) and  (\ref{nz-1a}) simplify to
\begin{eqnarray}\label{nz-2}
&&\hspace{-8mm}S_0(1-P_S)\big(1-e^{-G/\bar\Omega_R}\big)
=
\frac{1}{\ln(2)}\bigg[
P_S E_1\left(\frac{\lambda}{\bar\Omega_R}\right) 
\nonumber\\
&&\hspace{-6mm}
+(1-P_S) \left( E_1\left(\frac{G}{\bar\Omega_R}\right) +\ln\left(\frac{G}{\lambda}\right) e^{-G/\bar\Omega_R} \right)\bigg]
\end{eqnarray}
and
\begin{eqnarray}\label{nz-2a}
&&\hspace{-6mm}P_S\left[\frac{e^{-\lambda/\bar\Omega_R}}{\lambda} -\frac{1}{\bar\Omega_R} E_1\left(\frac{\lambda}{\bar\Omega_R}\right)\right] +(1-P_S) \bigg[\frac{e^{-G/\bar\Omega_R}}{\lambda} \nonumber\\
&&\hspace{-6mm}
-\frac{1}{\bar\Omega_R}E_1\left(\frac{G}{\bar\Omega_R}\right)\bigg] 
+\gamma_S(1-P_S)\big(1-e^{-G/\bar\Omega_R}\big)=\Gamma,\qquad
\end{eqnarray}
respectively, 
 where integral limit $G$ is given by (\ref{int-lim-1}). The maximum throughput is given by the right (and left) hand side of (\ref{nz-2}). 

For the case, where (\ref{cond-PS-mixed-pa-1-ray}) does not hold, the throughput is given by $\tau= S_0(1-P_S)$.
\end{lemma}
\begin{proof}
Equations (\ref{cond-PS-mixed-pa-1-ray}), (\ref{eq_un-1}),  (\ref{nz-2}),  and (\ref{nz-2a})
 are  obtained  by inserting the pdfs of $h_S(i)$ and $h_R(i)$ into (\ref{cond-PS-mixed-pa-1}), (\ref{lambda-cond-PS-mixed-pa-1}), 
 (\ref{nz-1}),  and (\ref{nz-1a}),
respectively.
\end{proof}

\begin{remark}
Conditions (\ref{cond-PS-mixed-non-pa-1}) and (\ref{cond-PS-mixed-pa-1})  depend only on the long term fading statistics and not on the instantaneous fading states. Therefore, for fixed $\bar\Omega_S$ and $\bar\Omega_R$, the optimal policy  for condition (\ref{cond-PS-mixed-non-pa-1}) is given by either (\ref{sol-d-mixed-1}) or (\ref{sol-d-mixed-2}), but not by both. Similarly, the optimal policy  for condition (\ref{cond-PS-mixed-pa-1}) is given by either (\ref{sol-d-mixed-PA}) or (\ref{sol-d-mixed-PA-2}), but not by both.
\end{remark}
\subsection{Mixed Rate Transmission with Delay Constraints}\label{sec-delay-2}
Now, we turn our attention to mixed rate transmission with delay constraints. For the delay unconstrained case, Theorem~\ref{theorem1-mixed} was very useful to arrive at the optimal protocol since it removed the complexity of having to deal with the queue states. 
However, for the delay constrained case, the queue states determine the throughput and the average delay.  Moreover, for mixed rate transmission, the queue states can only be modeled by a Markov chain with continuous state space, which makes the analysis 
complicated.  Therefore, we resort to a suboptimal adaptive link selection protocol in the following.
\begin{proposition}\label{prep-4}
Let the buffer size be limited to $Q_{\max}$ bits. For this case, we propose the following link selection protocol for mixed rate transmission with delay constraints: 
\begin{enumerate}
\item If $O_S(i)=0$, set $d_i=1$.
\item  Otherwise, if $\log_2(1+r(i))  \leq Q(i-1) \leq Q_{\max}-S_0$, select $d_i$ as proposed  in Theorem \ref{theorem7}   for the case of transmission without delay constraint.
\item Otherwise, if $Q(i-1)> Q_{\max}-S_0$, set $d_i=1$.
\item Otherwise, if  $Q(i-1)< \log_2(1+r(i))$, set $d_i=0$.
\end{enumerate}
\end{proposition}
If the $\mathcal{S}$-$\mathcal{R}$ link is in outage, the relay transmits. Otherwise, if  there is enough room in the buffer to accommodate the bits possibly sent from the source to the relay and there are enough bits in the buffer for the relay to transmit,  
the link selection protocol introduced in Theorem \ref{theorem7}  is employed. On the other hand, if there exists the possibility of a buffer overflow, the relay transmits to reduce the amount of data in the buffer. If the number of bits in the buffer 
is too low, the source transmits. The value of $Q_{\max}$ can be used to adjust the average delay while maintaining a low throughput loss compared to  the throughput without delay constraint. 

Although conceptually simple, as pointed out before, a theoretical analysis of the throughput of the proposed queue size limiting protocol is difficult because of the continuous state space of the associated Markov chain. Thus, we will resort to simulations to 
evaluate its performance in Section \ref{numerics}.
\subsection{Conventional Relaying With Delay Constraints \label{mixed-delay}}
To have a benchmark for delay constrained buffer-aided relaying with adaptive link selection, we propose a corresponding conventional relaying protocol, which may be viewed as a delay constrained version of Conventional Relaying 1.
\begin{proposition}\label{pr_3}
The source transmits to the relay in $k$ consecutive time slots followed by the relay transmitting to the destination in the following $n$ time slots. Then, this patter is repeated, i.e., the source transmits again in $k$ consecutive time slots, and so on. 
The values of $k$ and $n$ can be chosen to satisfy any delay and throughput requirements. 
\end{proposition}

For this protocol, the queue is non-absorbing if
\begin{eqnarray}\label{mix_delay_conv_cond}
    k(1-P_S)S_0\leq nE\{\log_2(1+r(i))\}.
\end{eqnarray}
Assuming (\ref{mix_delay_conv_cond}) holds, the average arrival rate is equal to the throughput and hence the  throughput is given by
\begin{eqnarray}  
\tau&=&\frac{k}{k+n}(1-P_S)S_0\label{eq_m_t_1}\;,
\end{eqnarray}

Using a numerical example, we will show in Section \ref{numerics} (cf.~Fig.~\ref{fig4}) that the protocol with adaptive link selection in Proposition~\ref{prep-4} achieves a higher throughput than the conventional protocol in Proposition~\ref{pr_3}. 
However, the conventional protocol is more amendable to analysis and it is interesting to investigate the corresponding throughput and multiplexing gain for a given average delay in the high SNR regime, $\gamma_S=\gamma_R=\gamma\to\infty$.
This is done in the following theorem.

\begin{theorem}\label{theorem9}
For a given average delay constraint, $E\{T\}$, the maximal throughput $\tau$ and multiplexing rate $r$ of mixed rate transmission, for $\gamma_S=\gamma_R=\gamma\to\infty$, are given by
\begin{eqnarray}
    \tau&\to& S_0 \left(1-\frac{1}{2 E\{T\}}\right),\quad \textrm{as }\gamma\to\infty\;. \label{t_b_N-1_2}\\
r&\to& 1-\frac{1}{2 E\{T\}}  ,\quad \textrm{as }\gamma\to\infty\;. \label{t_b_N-1_22121}
\end{eqnarray}
\end{theorem}
\begin{proof}
Please refer to Appendix~\ref{app_theorem9}.
\end{proof}
\begin{remark}
Theorem \ref{theorem9} reveals that, as expected from the discussion of the case without delay constraints, delay constrained mixed rate transmission approaches a multiplexing gain of one as the allowed average delay increases.
\end{remark}
\section{Numerical and Simulation Results}\label{numerics}
In this section, we evaluate the performance of the proposed fixed rate and mixed rate transmission schemes for Rayleigh fading. We also confirm some of our analytical results with computer simulations.   We note that our analytical results are valid for 
$N\to\infty$. For the simulations, $N$ has to be finite, of course, and we adopted $N=10^7$ in all simulations. Furthermore,  in the simulations for buffer-aided relaying without delay constraints, we neglected transient effects caused by the filling and emptying 
of the buffer at the beginning and the end of transmission. This allows us to verify the theoretical results for this idealized case, which constitute performance upper bounds for the delay constrained case. On the other hand, for the practical delay constrained 
case transient effects are taken into account in our simulations. In particular, we assume that the buffer is empty at the beginning of transmission and, once the source has ceased to transmit,  the relay transmits the queued information in its buffer until the 
buffer is empty. In this case, the simulated performance of the proposed protocols takes into account all transmitted bits. However, our results show that for the adopted value of $N$, transient effects (which are not included in our theoretical expressions, which
where derived for $N\to\infty$) do not have a noticeable impact of on the performance of the proposed delay constrained protocols and there is an excellent agreement between the simulated and theoretical performance results, cf.~Figs.~3-5. 
\subsection{Fixed Rate Transmission}
For fixed rate transmission, we evaluate the proposed link selection protocols for transmission with and without delay constraints. Throughout this section we assume that source and relay transmit with identical rates, i.e., $S_0=R_0$. 
\subsubsection{Transmission Without Delay Constraints}
In Fig.~\ref{fig2a}, we show  the ratio of the  throughputs achieved with the proposed buffer-aided relaying protocol with adaptive link selection and Conventional Relaying 1 as a function of the transmit SNR $\gamma_S=\gamma_R=\gamma$ for $\bar \Omega_R=1$,
$S_0=R_0=2$ bits/slot, and different values of $\bar\Omega_S$. The throughput of buffer-aided relaying, $\tau$, was computed based on (\ref{max-tau-1}), (\ref{max-tau-2}), and (\ref{max-tau-3}) in Theorem \ref{theorem2}, while the  throughput of Conventional 
Relaying 1, $\tau_{\rm conv,1}^{\rm fixed}$, was obtained based on (\ref{eq_qq-conv_1a}).  Furthermore, we also show simulation results where the throughput of the buffer-aided relaying protocol was obtained via Monte Carlo simulation. From Fig.~\ref{fig2a} we
observe that theory and simulation are in excellent agreement. Furthermore, Fig.~\ref{fig2a} shows that except for $\bar\Omega_S=\bar\Omega_R$ the proposed link adaptive relaying scheme achieves its largest gain for medium SNRs. For very high SNRs, both links are 
never in outage and thus, Conventional Relaying 1 with optimized $\xi$ and link adaptive relaying achieve the same performance. On the other hand, for very low SNR, there are very few transmission opportunities on both links as the links are in outage most of the time.
The proposed link adaptive protocol can exploit all of these opportunities. In contrast, for $\bar\Omega_S=\bar\Omega_R$, Conventional Relaying 1 choses $\xi=0.5$ and will miss half of the transmission opportunities by selecting the link that is in outage instead of the link
that is not in outage because of the pre-determined schedule for link selection. On the other hand, if $\bar\Omega_S$ and $\bar\Omega_R$ differ significantly, Conventional Relaying 1 selects $\xi$ close to 0 or 1 (depending on which link is stronger) and the loss 
compared to the link adaptive scheme becomes negligible. 

\begin{figure}
\includegraphics[width=3.75in]{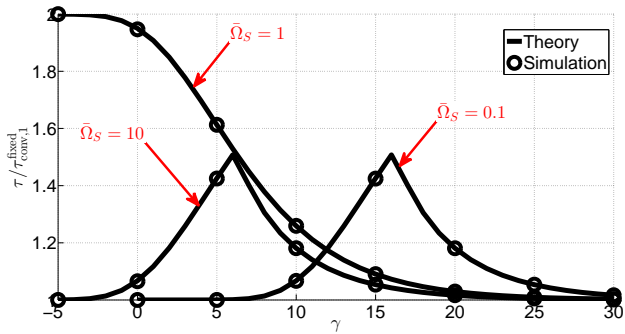}
\caption{Ratio of the throughputs of buffer-aided relaying and Conventional Relaying 1, $\tau/\tau_{\rm conv,1}^{\rm fixed}$, vs.~$\gamma$. Fixed rate transmission without delay constraints. $\gamma_S=\gamma_R=\gamma$, $S_0=R_0=2$ bits/slot, and $\bar \Omega_R=1$.} \label{fig2a}
\end{figure}
\begin{figure}
\includegraphics[width=3.75in]{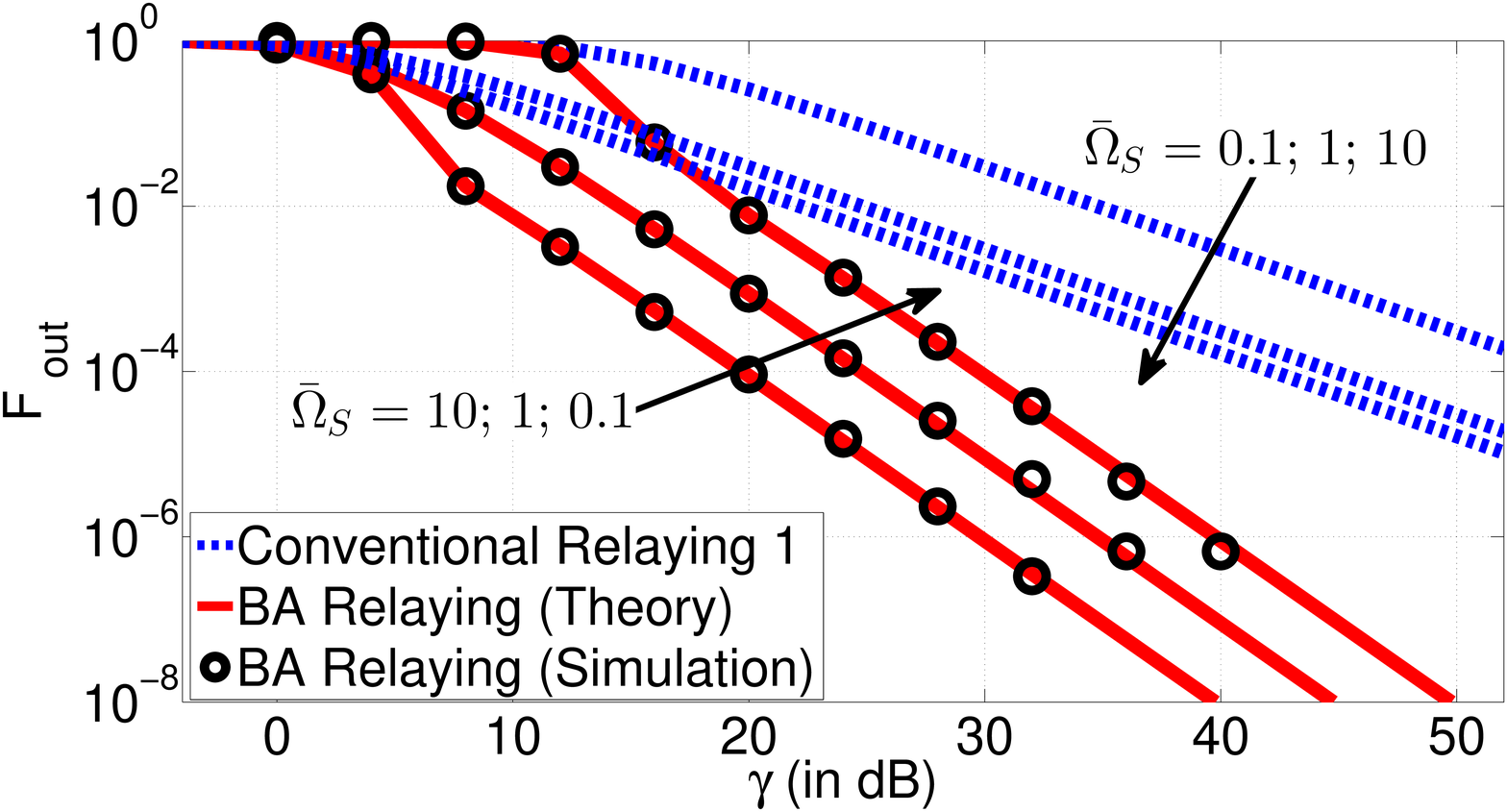}
\caption{Outage probability of buffer-aided (BA) relaying and Conventional Relaying 1 vs.~$\gamma$. Fixed rate transmission without delay constraints. $\gamma_S=\gamma_R=\gamma$, $S_0=R_0=2$ bits/slot, and $\bar \Omega_R=1$.} \label{fig2}
\end{figure}

In Fig.~\ref{fig2},  we show the outage probability, $F_{\rm out}$, for the proposed buffer-aided relaying protocol with adaptive link selection and Conventional Relaying 1. The same channel and system parameters as for Fig.~\ref{fig2a} were adopted
for Fig.~\ref{fig2} as well. For buffer-aided relaying with adaptive link selection,  $F_{\rm out}$ was obtained from  (\ref{OP-non-delay}) and confirmed by Monte Carlo simulations. For conventional relaying, $F_{\rm out}$ was obtained from (\ref{op_conv_1}). 
As expected from Lemma \ref{lemma_3}, buffer-aided relaying achieves a diversity gain of two, whereas conventional relaying achieves only a diversity gain of one, which underlines the superiority of buffer-aided relaying with adaptive link selection.
\subsubsection{Transmission With Delay Constraints}\label{sec-numerics-delay}
In Fig.~\ref{fig_delay_2a}, we show the throughput of buffer-aided relaying with adaptive link selection as  a function of the transmit SNR $\gamma_S=\gamma_R=\gamma$ for fixed rate transmission with different constraints on the average 
delay $E\{T\}$. The theoretical curves for buffer-aided relaying were obtained from the expressions given in Lemma \ref{lemma5} for throughput and the average delay. For comparison, we also show the throughput of buffer-aided relaying with adaptive link 
selection and without delay constraint (cf.~Theorem \ref{theorem2}), and the throughput of Conventional Relaying 2  given by (\ref{trup_conv_2}). These two schemes introduce an infinite delay, i.e., $E\{T\}\to\infty$ as $N\to\infty$, and a delay of one time slot, 
respectively. In the low SNR regime, the proposed buffer-aided relaying scheme with adaptive link selection cannot satisfy all delay requirements as expected from Lemma~\ref{lemma5a}. Hence, for finite delays, the throughput curves in 
Fig.~\ref{fig_delay_2a} do not extend to low SNRs. Nevertheless, as the affortable delay increases, the throughput for delay constrained transmission approaches the throughput for delay unconstrained transmission for sufficiently high SNR. 
Furthermore, the performance gain compared to Conventional Relaying 2 is substantial even for the comparatively small average delays $E\{T\}$ considered in Fig.~\ref{fig_delay_2a}.

In Fig.~\ref{fig_delay_2}, we show the outage probability, $F_{\rm out}$, for the same schemes and parameters that were considered in Fig.~\ref{fig_delay_2a}. For buffer-aided relaying with adaptive link selection,  
the theoretical results shown in Fig.~\ref{fig_delay_2} were obtained from (\ref{OP-delay-1-and-2})  and (\ref{OP-delay-3}). These theoretical results are confirmed by the Monte Carlo simulation results also shown in 
Fig.~\ref{fig_delay_2}. Furthermore, the curves for transmission without delay constraint (i.e., $E\{T\}\to\infty$ as $N\to\infty$) were computed from (\ref{OP-non-delay}), and for Conventional Relaying 2, we used (\ref{op_conv_2}). 
In addition, we have included in Fig.~\ref{fig_delay_2} the outage probability of the buffer-aided relaying protocol proposed in \cite[Section V.C]{globe11}. The results for the latter protocol were obtained via Monte Carlo simulation.
Fig.~\ref{fig_delay_2} shows that even for an average delay as small as $E\{T\}=1.1$ slots, the proposed buffer-aided relaying protocol with adaptive link selection outperforms Conventional Relaying 2. Furthermore,  as expected from 
Theorem \ref{theo-outage-high-snr}, buffer-aided relaying with adaptive link selection achieves a diversity gain of two when the average delay is larger than three time slots (e.g.,~$E\{T\}=3.1$ time slots in Fig.~\ref{fig_delay_2} ). 
This leads to a large performance gain over conventional relaying which achieves only a diversity gain of one. Finally, note that even for  $E\{T\}=3.1$ the coding gain loss is very small compared to the case of $E\{T\}\to\infty$. This is in
 stark contrast to the protocol proposed in  \cite[Section V.C]{globe11}, which suffers from a loss in diversity even for an average delay of $E\{T\}=50$.

\begin{figure*}[!t]
\centering
\includegraphics[width=6in]{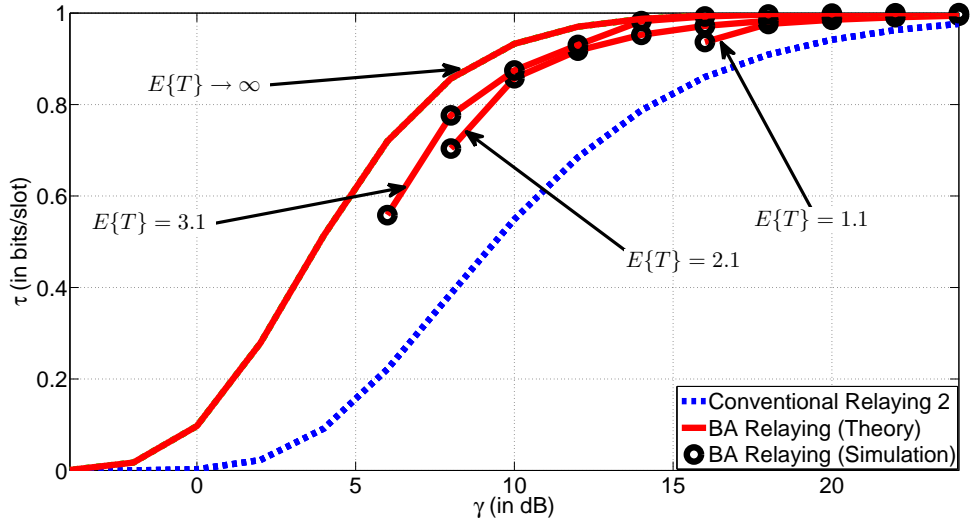}
\caption{Throughputs of buffer-aided  (BA) relaying and Conventional Relaying 2 vs.~$\gamma$. Fixed rate transmission with delay constraints.  $\gamma_S=\gamma_R=\gamma$, $S_0=R_0=2$ bits/slot, $\bar \Omega_R=1$, 
and $\bar\Omega_S=1$.} \label{fig_delay_2a}
\end{figure*}
\begin{figure*}[!t]
\includegraphics[width=6in]{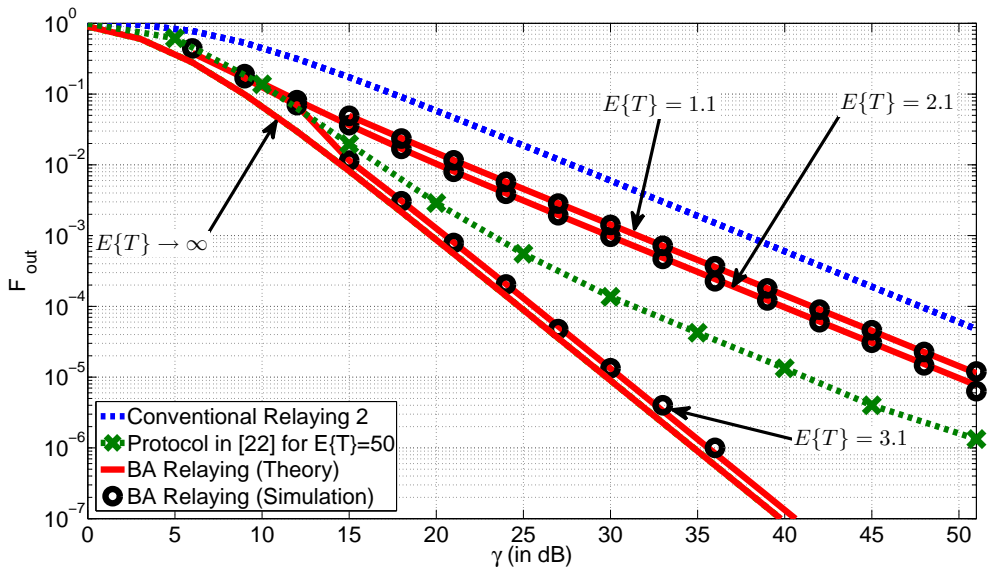}
\centering
\caption{Outage probability of buffer-aided (BA) relaying, Conventional Relaying 2, and the protocol proposed in \cite[Section V.C]{globe11} vs.~$\gamma$. Fixed rate transmission with delay constraints.  $\gamma_S=\gamma_R=\gamma$, $S_0=R_0=2$ bits/slot, $\bar \Omega_R=1$, 
and $\bar\Omega_S=1$.} \label{fig_delay_2}
\end{figure*}
 
\begin{remark}
For the simulation results shown in Figs.~\ref{fig_delay_2a} and \ref{fig_delay_2}, we adopted a relay with a buffer size of $L=60$ packets which leads to a negligible probability of dropped packets. For example, for $\gamma=45$ dB, the probability of a 
full buffer, ${\rm Pr}\{Q=LR_0  \}$, is bounded by ${\rm Pr}\{Q=LR_0  \}<10^{-60}$. This also supports the claim in the proof of Theorem \ref{theo-outage-high-snr} that for  large enough buffer sizes the probability of dropping a packet due to a buffer overflow
becomes negligible.
\end{remark}
\subsection{Mixed Rate Transmission}
In this section, we investigate the achievable throughput for mixed rate transmission. For this purpose, we consider again the delay constrained and the delay unconstrained cases separately. 
\subsubsection{Transmission Without Delay Constraints}
In Fig.~\ref{fig-pa}, we compare the throughputs of buffer-aided relaying with adaptive link selection and Conventional Relaying 1. In both cases, we consider the cases with and without power allocation. The theoretical results shown
in Fig.~\ref{fig-pa} for the four considered schemes were generated based on Theorem~\ref{theorem7}/Lemma~\ref{lemma7}, Theorem~\ref{theorem8}/Lemma~\ref{lemma8}, (\ref{eq_qq-1}), (\ref{t-conv-mixed-2}), and 
(\ref{eq_qq-1}), (\ref{t-conv-mixed-2pa}). The transmit SNRs of both links are identical, i.e., $\gamma_S=\gamma_R=\Gamma$, $S_0=2$ bits/slot, $\bar\Omega_S=10$, and $\bar\Omega_R=1$.  As can be observed from 
Fig.~\ref{fig-pa}, for both buffer-aided relaying with adaptive link selection and Conventional Relaying 1, power allocation is beneficial only for low to moderate SNRs. Both schemes can achieve a throughput of $S_0$ 
bits/slot in the high SNR regime. However, adaptive link selection achieves a throughput gain compared to Conventional Relaying 1 in the entire considered SNR range.
\begin{figure*}
\centering
\includegraphics[width=6in]{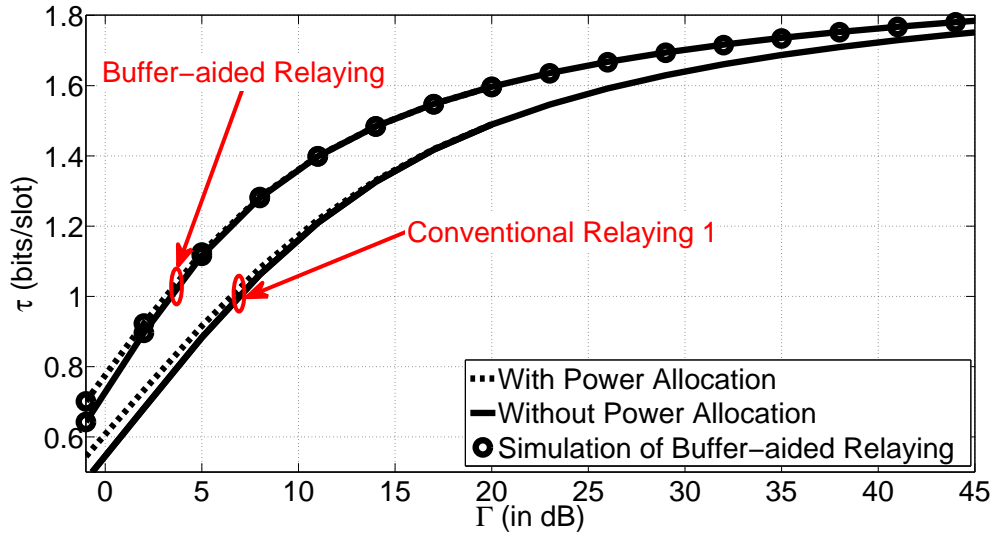}
\caption{Throughput of buffer-aided relaying with adaptive link selection and  Conventional Relaying  1 vs.~$\Gamma$. Mixed rate transmission without delay constraints. $\bar\Omega_S=10$, $\bar\Omega_R=1$, and $S_0=2$ bits/slot.}  \label{fig-pa}
\end{figure*}

\subsubsection{Transmission with Delay Constraints}
\begin{figure*}[!t]
\includegraphics[width=6in]{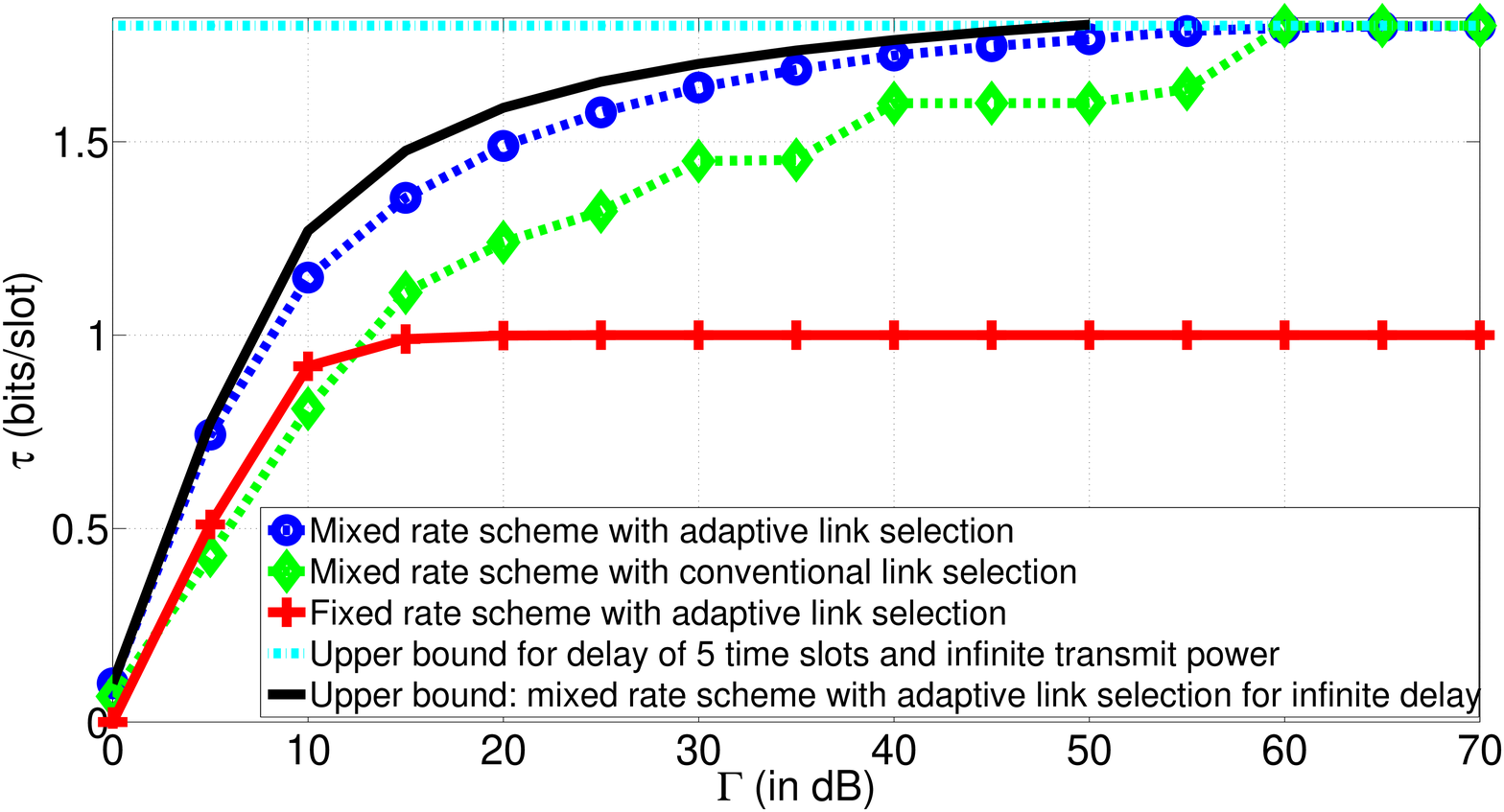}
\centering
\caption{Throughput  of buffer-aided relaying with adaptive link selection and  conventional relaying  vs.~$\Gamma$. Mixed rate and fixed rate transmission with  delay constraint.  $E\{T\}=5$ time slots, $\gamma_S=\gamma_R=\Gamma$, $S_0=2$ bits/slot, 
and $\bar\Omega_S=\bar\Omega_R=1$.}  \label{fig4}
\vspace*{-1mm}
\end{figure*}
In Fig.~\ref{fig4}, we compare the throughputs of various mixed rate and fixed rate transmission schemes for a maximum average delay of $E\{T\}=5$ time slots and $S_0=2$ bits/slot.  The transmit SNRs of both links are identical, i.e., $\gamma_S=\gamma_R=\Gamma$,  $\bar\Omega_S= \bar\Omega_R=1$.
For mixed rate transmission, we simulated both the buffer-aided relaying protocol with adaptive link selection described in Proposition~\ref{prep-4} and the conventional relaying protocol described in Proposition~\ref{pr_3}.
For fixed rate transmission, we chose $R_0=S_0=2$ bits/slot and included results for buffer-aided relaying with adaptive link selection obtained based on Lemma~\ref{lemma5}. Furthermore, for mixed rate transmission, we also show 
the maximum achievable throughput of buffer-aided relaying with adaptive link selection in the absence of delay constraints (as given by Theorem~\ref{theorem7}/Lemma~\ref{lemma7}) and the maximum throughput achievable for a delay 
constraint of $E\{T\}=5$ time slots and infinite transmit power (as given by (\ref{t_b_N-1_2})). Fig.~\ref{fig4} reveals that for mixed rate transmission the protocol with adaptive link selection proposed in Proposition~\ref{prep-4} is superior to  
the conventional relaying scheme proposed in Proposition~\ref{pr_3}, and for high SNR, both protocols reach the upper bound for mixed rate transmission under a delay constraint  given by (\ref{t_b_N-1_2}). Furthermore, Fig.~\ref{fig4} also shows that mixed rate transmission is superior
to fixed rate transmission since the former can exploit the additional flexibility afforded by having CSIT for the $\cal R$-$\cal D$ link. For example, for $\Gamma=30$ dB, mixed rate transmission with adaptive link selection achieves a throughput 
gain of 65 $\%$  compared to fixed rate transmission, and even conventional link selection  still achieves a gain of $45 \%$. Fig.~\ref{fig4} also shows that even in the presence of severe delay constraints mixed rate transmission can significantly reduce 
the throughput loss caused by half-duplexing compared to full-duplexing, whose maximum throughput is $S_0=2$ bits/slot.\footnote{We note that for transmitting and receiving in the same time slot and the same frequency band, a full-duplex relay would need  two antennas, one for transmission and one for reception \cite{jain2011practical}, whereas the half-duplex relay considered in this paper only requires one antenna which can be used for reception and transmission in different time slots.
However, a decode-and-forward full-duplex relay can retransmit the packet received in the current time slot in the following time slot and has to store it only for one time slot.}
\section{Conclusions}\label{conclude}
In this paper, we have considered a three-node decode-and-forward relay system comprised of a source, a half-duplex relay with a buffer, and a destination, where the direct source-destination link is not available or not used. 
We have investigated both fixed rate transmission, where
source and relay do not have full CSIT and are forced to transmit with fixed rate, and mixed rate transmission, where the source does not have full CSIT and transmits with fixed rate but the relay has full CSIT and transmits with variable rate. For both modes of transmission,
we have derived the throughput-optimal buffer-aided relaying protocols with adaptive link selection and the resulting throughputs and outage probabilities. Furthermore, we could show that buffer-aided relaying with adaptive link selection leads to
substantial performance gains compared to conventional relaying with non-adaptive link selection. In particular, for fixed rate transmission, buffer-aided relaying with adaptive link selection achieves a diversity gain of two, whereas conventional relaying is 
limited to a diversity gain of one. For mixed rate transmission, both buffer-aided relaying with adaptive link selection and a newly proposed conventional relaying scheme with non-adaptive link selection have been shown to overcome the half-duplex loss typical
for wireless relaying protocols and to achieve a multiplexing gain of one. Since the proposed throughput-optimal protocols introduce an infinite delay, we have also proposed modified protocols for delay constrained transmission and have investigated the resulting 
throughput-delay trade-off. Surprisingly, the diversity gain of fixed rate transmission with buffer-aided relaying is also observed for delay constrained transmission as long as the average delay exceeds three time slots. Furthermore, for mixed rate 
transmission, for an average delay $E\{T\}$, a multiplexing gain of $r=1-1/(2E\{T\})$ is achieved even for conventional relaying. 
\begin{appendix}


\subsection{Proof of Theorem \ref{theorem1}\label{app_for_t_1}}
We first note that, because of the law of the conservation of flow, $A\ge \tau$ is always valid and equality holds if and only if the queue is non-absorbing. 

We denote the set of indices with $d_i=1$ by $\bar I$ and the set of indices with $d_i=0$ by $I$.
Assume   that we have a link selection protocol with arrival rate $A$ and throughput $\tau$ with $A>\tau$, 
i.e., the queue is absorbing. Then, for $N\to \infty$, we have
\begin{eqnarray}
A&=&\frac{1}{N}\sum_{i\in I}(1-d_i) O_S(i) S_0 \nonumber\\
&>& \tau =\frac{1}{N}\sum_{i\in\bar I}  d_i O_R(i) \min\{R_0,Q(i-1)\}.
\label{app1}
\end{eqnarray}
From (\ref{app1}) we observe that the considered protocol cannot be optimal as the throughput can be improved by moving some of the indices $i$ in $I$ to $\bar I$ which leads to an increase of $\tau$ at the expense of a decrease of $A$. 
As we continue moving indices from $I$ to $\bar I$ we reach a point where $A=\tau$ holds. At this point, the queue becomes non-absorbing (but is at the boundary between a non-absorbing and an absorbing queue) and the throughput is
maximized. If we continue moving indices from $I$ to $\bar I$, in general, $A$ will decrease and as a consequence of the law of conservation of flow, $\tau$ will also decrease. We note that $A$ does not decrease if we move only those indices 
from $I$ to $\bar I$ for which $O_S(i)=0$ holds. In this case, $A$ will not change, and as a consequence of the law of conservation of flow, the value of $\tau$ also remains unchanged.
Note that this is used in Lemma~1.
 However, the queue is moved from the edge of non-absorption 
if $O_R(i)=1$ holds for some of the indices moved from $I$ to $\bar I$. As will be seen later, if the queue of the buffer operates at the edge of non-absorption, the throughput becomes independent of the state of the queue, which is desirable 
for analytical throughput maximization. 

In the following, we will prove that when the queue is at the edge of non-absorption the following holds
\begin{eqnarray}
   &&\hspace{-6mm} \tau= \lim_{N\to \infty} \frac{1}{N}\sum_{i=1}^N d_i O_R(i) R_0 \nonumber\\
&&\hspace{-6mm}= A= \lim_{N\to \infty} \frac{1}{N}\sum_{i=1}^N (1-d_i) O_S(i) S_0.
\end{eqnarray}

Let $\epsilon$ denote a small subset of $\bar I$ containing only indices $i$ for which $O_S(i)=1$, where $|\epsilon|/N\to 0$ for $N\to\infty$ and $|\cdot|$ denotes the cardinality of a set. Throughout the remainder of this proof $N\to\infty$ is assumed.

If the queue in the buffer of the relay is absorbing, $A>\tau$ holds and on average the number of bits arriving at the queue exceed the number of bits leaving the queue. Thus, $R_0\le Q(i-1)$ holds almost always and as a result the throughput can be written as 
\begin{equation}\label{pr_eq_2}
   \tau= \frac{1}{N}\sum_{i\in \bar I } O_R(i) \min\{R_0,Q(i-1)\} = \frac{1}{N} \sum_{i\in \bar I} O_R(i) R_0 .
\end{equation}

Now,  we assume that the queue is at the edge of non-absorption. That is $A=\tau$ holds but  moving the small fraction of indices in $\epsilon$, where $|\epsilon|/N\to 0$, from $\bar I$ to $I$ will make  the queue an absorbing queue with $A>\tau$. 
For this case, we wish to determine whether or not
\begin{eqnarray}\label{pr_eq_3}
  &&\hspace{-6mm} \frac{1}{N} \sum_{i\in \bar I} O_R(i) R_0> \tau =\frac{1}{N}\sum_{i\in \bar I } O_R(i) \min\{R_0,Q(i-1)\}\nonumber\\
&&\hspace{-6mm}= A=\frac{1}{N} \sum_{i\in   I }  O_S(i) S_0
\end{eqnarray}
holds. To test this, we move a small fraction $\epsilon$, where $|\epsilon|/N\to 0$, of indices from $\bar I$ to $I$, thus making the queue an absorbing queue. As a result, (\ref{pr_eq_2}) holds and (\ref{pr_eq_3}) becomes
\begin{eqnarray}\label{pr_eq_4}
  &&\hspace{-6mm} \frac{1}{N} \sum_{i\in \bar I\backslash\epsilon} O_R(i)  R_0= \tau =\frac{1}{N}\sum_{i\in \bar I\backslash \epsilon } O_R(i) \min\{R_0),Q(i-1)\}\nonumber\\
&&\hspace{-6mm}< A= \frac{1}{N}\sum_{i\in I \cup \epsilon } O_S(i) S_0.
\end{eqnarray}
From the above we conclude that if (\ref{pr_eq_2}) holds, then based on (\ref{pr_eq_3}) and (\ref{pr_eq_4}), for $|\epsilon|/N\to 0$, we must have
\begin{eqnarray}\label{pr_eq_5}
   \frac{1}{N} \sum_{i\in\bar I} O_R(i) R_0>  \frac{1}{N}\sum_{i\in   I } O_S(i) S_0
\end{eqnarray}
and
\begin{eqnarray}\label{pr_eq_6}
   \frac{1}{N} \sum_{i\in \bar I\backslash\epsilon} O_R(i) R_0 <  \frac{1}{N}\sum_{i\in   I \cup \epsilon } O_S(i) S_0.
\end{eqnarray}
However, for (\ref{pr_eq_5}) and (\ref{pr_eq_6}) to jointly hold, we require that the particular considered move of indices from $\bar I$ to $I$ causes a discontinuity in $\frac{1}{N} \sum_{i\in \bar I} O_R(i) R_0$ or/and a discontinuity in 
$\frac{1}{N}\sum_{i\in   I } O_S(i) S_0$ as $|\epsilon|/N\to 0$ is assumed. Since $S_0$ and $R_0$ are finite,  $\lim_{N\to\infty}\sum_{i\in \epsilon }S_0/N=\lim_{N\to\infty}S_0|\epsilon|/N= 0$ and $\lim_{N\to\infty}\sum_{i\in \epsilon }R_0/N=\lim_{N\to\infty}R_0|\epsilon|/N= 0$. 
Hence, such discontinuities are not possible. Therefore, at the edge of non-absorption the inequality in (\ref{pr_eq_3}) cannot hold and we must have
\begin{eqnarray}\label{pr_eq_7}
 &&\hspace{-6mm}  \frac{1}{N} \sum_{i\in \bar I} O_R(i) R_0= \tau= \frac{1}{N}\sum_{i\in \bar I } O_R(i) \min\{R_0,Q(i-1)\}\nonumber\\
&&\hspace{-6mm}=A= \frac{1}{N}\sum_{i\in   I } O_S(i) S_0.
\end{eqnarray}
Eq.~(\ref{pr_eq_7}) can be written as (\ref{cond-1}). This concludes the proof. 
\subsection{Proof of Theorem \ref{theorem2}\label{app_A}}
The Lagrangian of Problem (\ref{MPR1}) is given by
\begin{eqnarray}\label{MPR2}
    \mathcal{L}&&\hspace{-6mm}=\frac{1}{N}\sum_{i=1}^N d_i O_R(i) R_0 - \sum_{i=1}^N \tilde{\beta}_i d_i(1-d_i) \nonumber\\
&&\hspace{-6mm}-\mu \frac{1}{N}\sum_{i=1}^N  \Big[d_i O_R(i) R_0-(1-d_i) O_S(i) S_0\Big] ,
\end{eqnarray}
where $\mu$ and $\tilde{\beta}_i$ are the Lagrange multipliers. Differentiating $\mathcal{L}$ with respect to $d_i$, introducing $\beta_i=N\tilde{\beta}_i$, and setting the result to zero leads to
\begin{eqnarray}\label{eq-nekoja1}
    d_i=\frac{\beta_i+(-1+\mu)O_R(i) R_0+\mu O_S(i) S_0}{2\beta_i}.
\end{eqnarray}
For $d_i(1-d_i)=0$ to hold, we need either $d_i=0$ or $d_i=1$, which leads to two possible values for $\beta_i$:
\begin{eqnarray}
d_i=0\, \Rightarrow\,    \beta_{i,1}&=&(1-\mu)O_R(i) R_0-\mu O_S(i) S_0\\
d_i=1\, \Rightarrow\,  \beta_{i,2}&=&-\beta_{i,1}.
\end{eqnarray}
For the maximum of $\mathcal{L}$ in (\ref{MPR2}), $\beta_i\leq 0$, $\forall i$, has to hold.
Hence, we have
\begin{eqnarray}\label{link-select-1}
   d_i=\left\{
\begin{array}{cl}
1 & \textrm{if }  (1-\mu) O_R(i) R_0\geq \mu O_S(i) S_0\\
0 & \textrm{if }  (1-\mu) O_R(i) R_0\leq \mu O_S(i) S_0 .
\end{array} 
\right.
\end{eqnarray}
 Furthermore, $0\leq \mu \leq 1$ has to hold since for $\mu < 0$ and $\mu > 1$ we have always $d_i=1$ and $d_i=0$, respectively, irrespective of any non-negative values of $O_S(i) S_0$ and $O_R(i) R_0$. 

First, we consider the case $0 < \mu < 1$. The boundary values $\mu=0$ and $\mu=1$  will be investigated later. From  (\ref{link-select-1}), for $0 < \mu < 1$, we have four possibilities:
\begin{enumerate}
\item If $O_R(i)=1$ and $O_S(i)=0$, then $d_i=1$.
\item If $O_R(i)=0$ and $O_S(i)=1$, then $d_i=0$.
\item If $O_R(i)=0$ and $O_S(i)=0$, then $d_i$ can be chosen to be either $d_i=0$ or $d_i=1$ and the choice does not influence the throughput as both the source and the relay remain silent. 
\item If $O_R(i)=1$ and $O_S(i)=1$ and $\mu$ is chosen such that $0<\mu<R_0/(S_0+R_0)$ then $d_i=1$  in all time slots with $O_R(i)=1$ and $O_S(i)=1$,  and as a result, condition C1 cannot be satisfied. Similarly, if  $\mu$ is chosen such that $R_0/(S_0+R_0)<\mu<1$,
then $d_i=0$ in all time slots with $O_R(i)=1$ and $O_S(i)=1$,  and as a result condition C1 can also not be satisfied. Thus, we conclude that $\mu$ must be set to  $\mu=R_0/(S_0+R_0)$ since only in this case can $d_i$ be chosen to be either $d_i=0$ or $d_i=1$, which is
necessary for satisfying condition C1. Since for $O_R(i)=1$ and $O_S(i)=1$ neither link is in outage, $d_i$  can be chosen to be either zero or one, as long as condition C1 is satisfied. In order to satisfy C1, we propose to flip a coin and the outcome of the coin toss decides
whether $d_i=1$ or $d_i=0$. Let the coin have two outcomes $\mathcal{C}\in\{0,1\}$ with probabilities ${\rm  Pr}\{\mathcal{C}=0\}$ and ${\rm  Pr}\{\mathcal{C}=1\}$. We set  $d_i=0$ if $\mathcal{C}=0$ and $d_i=1$ if $\mathcal{C}=1$. Thus, the probabilities ${\rm  Pr}\{\mathcal{C}=0\}$ 
and ${\rm  Pr}\{\mathcal{C}=1\}$ have to be chosen such that  C1 is satisfied. 
\end{enumerate}
Choosing the link selection variable as in (\ref{sol-d-1}) and exploiting the independence of $s(i)$ and $r(i)$, condition C1 results in
\begin{eqnarray}\label{eqq-app-1}
&&\hspace{-6mm} S_0\left[ (1-P_S) P_R +(1-P_S)(1-P_R) {\rm  Pr}\{\mathcal{C}=0\}\right]\nonumber\\
&&\hspace{-6mm}
  = R_0 \left[(1-P_R) P_S +(1-P_S)(1-P_R) {\rm  Pr}\{\mathcal{C}=1\} \right].\qquad
\end{eqnarray}
From (\ref{eqq-app-1}), we can obtain the probabilities ${\rm  Pr}\{\mathcal{C}=0\}$ and ${\rm  Pr}\{\mathcal{C}=1\}$, which after some basic algebraic manipulations leads to (\ref{find-P_C-1}). The throughput is given by the right (or left) hand side of (\ref{eqq-app-1}),
which leads to (\ref{max-tau-1}).

For (\ref{find-P_C-1}) to be valid, ${\rm  Pr}\{\mathcal{C}=0\}$ and ${\rm  Pr}\{\mathcal{C}=1\}$ have to meet $0 \leq{\rm  Pr}\{\mathcal{C}=0\}\leq 1$ and $0 \leq{\rm  Pr}\{\mathcal{C}=1\}\leq 1$, which leads to the conditions 
\begin{eqnarray}
  S_0 (1-P_S)-(1-P_R)P_S R_0&\geq& 0\label{cond-001}\\
 R_0 (1-P_R)-(1-P_S)P_R S_0&\geq& 0  \label{cond-002}.
\end{eqnarray}
Solving (\ref{cond-001}) and (\ref{cond-002}), we obtain that for the link selection variable $d_i$ given in (\ref{sol-d-1}) to be valid, condition (\ref{cond-PS-PR-1}) has to be fulfilled. 

Next, we consider the case where $\mu=0$. Inserting $\mu = 0$ in (\ref{link-select-1}), we obtain three possible cases:
\begin{enumerate}
\item If $O_R(i)=1$, then $d_i=1$.
\item If $O_R(i)=0$ and $O_S(i)=0$, then $d_i$ can be chosen to be either $d_i=0$ or $d_i=1$ and the choice has no influence on the throughput. 
\item If $O_R(i)=0$ and $O_S(i)=1$, then $d_i$ can be chosen to be either $d_i=0$ or $d_i=1$ as long as condition C1 is satisfied. Similar to before, in order to satisfy C1, we propose to flip a 
coin and the outcome of the coin flip determines whether $d_i=1$ or $d_i=0$. 
\end{enumerate}
Choosing the link selection variable as in (\ref{sol-d-2}) and exploiting the independence of $s(i)$ and $r(i)$, condition C1 can be rewritten as
\begin{eqnarray}\label{eqq-app-2}
S_0 P_R(1-P_S) {\rm Pr}\{\mathcal{C}=0\} =R_0(1-P_R).
\end{eqnarray}
After basic manipulations (\ref{eqq-app-2}) simplifies to (\ref{find-P_C-2}). The throughput is given by the right (or left) hand side of (\ref{eqq-app-2}) and can be simplified to (\ref{max-tau-2}). Imposing again the conditions
$0 \leq{\rm  Pr}\{\mathcal{C}=0\}\leq 1$ and $0 \leq{\rm  Pr}\{\mathcal{C}=1\}\leq 1$, we find that for $\mu=0$, (\ref{cond-001}) still has to hold but (\ref{cond-002}) can be violated, which is equivalent to the new condition
\begin{eqnarray}
   P_R> \frac{R_0}{R_0+S_0(1-P_S)}.
\end{eqnarray}

For the third and final case, letting $\mu=1$ and following a similar path as for $\mu=0$ leads to (\ref{sol-d-3})--(\ref{max-tau-3}) and condition (\ref{cond-PS-PR-3}).

Finally, we have to prove that the three considered cases are mutually exclusive, i.e., for any combination of $P_S$ and $P_R$ only one case applies. Considering (33), (37), and (41) it is obvious that Cases 1 and 2 and Cases 1 and 3
are mutually exclusive, respectively. For Cases 2 and 3, the mutual exclusiveness is less obvious. Thus, we rewrite (37) and (41) as
\begin{equation}
P_R>P_{R,2} 
\label{eq37a}
\end{equation}
and 
\begin{equation}
P_R<P_{R,3},
\label{eq41a}
\end{equation}
respectively, where $P_{R,2}= R_0/(R_0+S_0(1-P_S))$ and $P_{R,3}=1+S_0/R_0-S_0/(R_0P_S)$. It can be shown that $P_{R,2}>P_{R,3}$ for any $0\le P_S<1$. Hence, for $0\le P_S<1$, at most one of (\ref{eq37a}) and (\ref{eq41a}) is satisfied and
Cases 2 and 3 are mutually exclusive. For $P_S=1$ (i.e., the ${\cal S}$-${\cal R}$ link is always in outage), we have $P_{R,2}=P_{R,3}=1$ and Case 1 and Case 3 apply for $P_R=1$ and $P_R<1$, respectively. Therefore, for any combination of 
$P_S$ and $P_R$ only one of the three cases considered in Theorem 2 applies.
This concludes the proof.
\subsection{Proof of Lemma \ref{lemma-OP-non-delay}\label{app_B}}
We provide two different proofs for the outage probability, $F_{\rm out}$, in (\ref{OP-non-delay}). The first proof is more straightforward and based on (\ref{OP-main}). However,  the second proof provides more insight 
into when outages occur. 

\textit{Proof 1:} In the absence of outages, the maximum achievable throughput, denoted by $\tau_0$, is given by (\ref{t_no_outage}). Thus, when (\ref{cond-PS-PR-1}) holds,  $F_{\rm out}$ is obtained by inserting (\ref{max-tau-1}) and (\ref{t_no_outage}) into (\ref{OP-main}).  
Similarly, when (\ref{cond-PS-PR-2}) holds,  $F_{\rm out}$ is obtained by inserting (\ref{max-tau-2}) and (\ref{t_no_outage}) into (\ref{OP-main}). Finally, when (\ref{cond-PS-PR-3}) holds,  $F_{\rm out}$ is obtained by inserting (\ref{max-tau-3}) and (\ref{t_no_outage}) into 
(\ref{OP-main}). After basic simplifications, (\ref{OP-non-delay}) is obtained. This concludes the proof.

\textit{Proof 2:}  The second proof exploits the fact that an outage occurs when both the source and the relay are silent, i.e., when none of the links is used.
When (\ref{cond-PS-PR-1}) holds, from $d_i$ given by (\ref{sol-d-1}), we observe that no transmission occurs only when both links are in outage. This happens 
with probability $F_{\rm out}=P_S P_R$. In contrast, when (\ref{cond-PS-PR-2}) holds,  from $d_i$ given by (\ref{sol-d-2}), we observe that no node transmits
when both links are in outage or when the $\mathcal{S}$-$\mathcal{R}$ link is not in outage, while the $\mathcal{R}$-$\mathcal{D}$ link is in outage and the 
coin flip chooses the relay for transmission. This event happens with probability $F_{\rm out}=P_S P_R+(1-P_S) P_R P_C$, which after inserting $P_C$ given 
by (\ref{find-P_C-2}) leads to (\ref{OP-non-delay}). Finally, when (\ref{cond-PS-PR-3}) holds, from $d_i$, given by (\ref{sol-d-3}), we see that no node transmits 
when both links are in outage or when the $\mathcal{S}$-$\mathcal{R}$  link is in outage, while the   $\mathcal{R}$-$\mathcal{D}$ link is not in outage and the 
coin flip chooses the source for transmission. This happens with probability $F_{\rm out}=P_S P_R+P_S (1-P_R)(1- P_C)$, which after introducing $P_C$ given 
by (\ref{find-P_C-3}) leads to (\ref{OP-non-delay}).
\subsection{Proof of Lemma \ref{lemma_3}}\label{app_F}
Computing the link outages in (\ref{PS}) and (\ref{PR}) for Rayleigh fading and exploiting (\ref{prob-eq}), we obtain (\ref{prob-eq-1-high-snr}) by employing $\Omega_S=\gamma\bar\Omega_S$ and $\Omega_R=\gamma\bar\Omega_R$ in the 
resulting expression and using a Taylor series expansion for $\gamma\to\infty$. As can be seen from (\ref{prob-eq-1-high-snr}), the transmit SNR $\gamma$ has an exponent of $-2$. Thus, the diversity order is two.

Moreover, for $\bar\Omega_S=\bar\Omega_R=\bar\Omega$ and $S_0=R_0$, the asymptotic expression for $F_{\rm out}$ in (\ref{prob-eq-1-high-snr}) simplifies to
\begin{eqnarray}\label{pdc}
F_{\rm out}\to\frac{(2^{R_0}-1)^2}{\bar\Omega^2 \gamma^2}, \quad\textrm{as }\gamma\to\infty .
\end{eqnarray}

Furthermore,  for $S_0=R_0$, the asymptotic throughput in (\ref{max-tau-1-high-snr}) simplifies to $\tau=R_0/2$. Thus, letting $\tau=r \log_2(1+\gamma)$ we obtain $R_0=2 r \log_2(1+\gamma)$. Inserting $R_0=2 r \log_2(1+\gamma)$ into (\ref{pdc}), the
diversity-multiplexing trade-off, $DM(r)$, is obtained as
\begin{eqnarray}
  &&\hspace{-6mm}  DM(r)=-\lim_{\gamma\to\infty} \frac{\log_2(F_{\rm out})}{\log_2(\gamma)}\nonumber\\
&&\hspace{-6mm}=-\lim_{\gamma\to\infty} \frac{2\log_2(2^{2r \log_2(1+\gamma)}-1)-2\log_2(\bar\Omega)-2\log_2(\gamma)}{\log_2(\gamma)}\nonumber\\
&&\hspace{-6mm}=2-\lim_{\gamma\to\infty} \frac{2\log_2((1+\gamma)^{2r}-1)}{\log_2(\gamma)}=2-4 r\;.
\end{eqnarray}
This completes the proof.

\subsection{Proof of Theorem \ref{theorem5}}\label{proof-th-5}
Let $d_i$ be given by $(\ref{sol-d-delay-1})$.  Then, the following events are possible for the queue in the buffer:
\begin{enumerate}
\item If the buffer is empty, it stays empty with probability $P_S$ and receives one packet with probability $1-P_S$.
\item If the buffer contains one packet, it stays in the same state with probability $P_S P_R$, sends the packet with probability $P_S(1-P_R)$, and receives a new packet with probability $1-P_S$.
\item If the buffer contains more than one  packet but less than $L$ packets, it stays in the same state with probability $P_S P_R$, receives a new packet with probability $(1-P_S)P_R+(1-P_S)(1-P_R)(1-P_C)$, and sends one packet 
with probability $(1-P_R)P_S+(1-P_S)(1-P_R) P_C$.
\item If the buffer contains  $L$ packets, it stays in the same state with probability $P_S P_R+ (1-P_S)P_R+(1-P_S)(1-P_R)(1-P_C)$, and sends one packet with probability $(1-P_R)P_S+(1-P_S)(1-P_R) P_C$. 
\end{enumerate}
\begin{figure}
\includegraphics[width=3.5in]{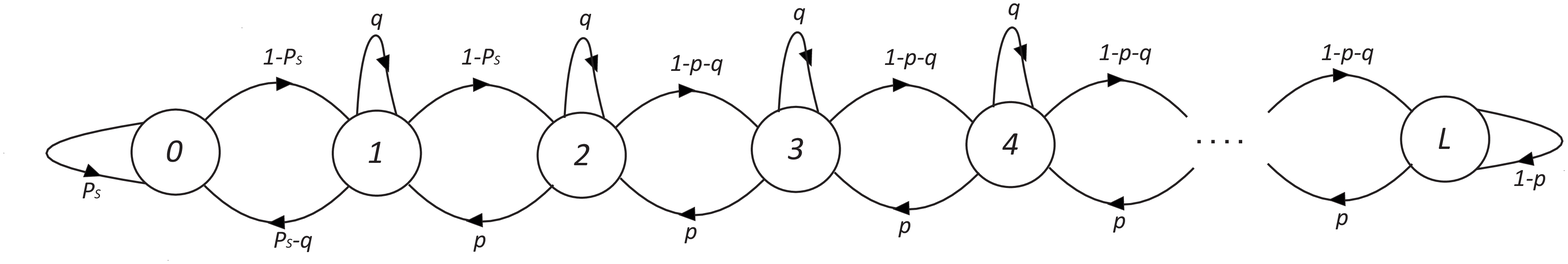}
\caption{Markov chain for the number of packets in the queue of the buffer if the link selection variable $d_i$ is given by $(\ref{sol-d-delay-1})$.}\label{markov-chain-1}
\end{figure}
The events for the queue of the buffer detailed above, form a Markov chain whose states are defined by the number of packets in the queue. This Markov chain is shown in Fig.~\ref{markov-chain-1}, where the probabilities $p$ and $q$ are given by 
(\ref{p&q-1}). Let $\mathbf M$ denote the state transition matrix of the Markov chain and let $m_{i,j}$ denote the element in the $i$-th row and $j$-th column of $\mathbf M$. Then, $m_{i,j}$ is the probability that the buffer will transition from having $i-1$ 
packets in its queue in the previous time slot to having $j-1$ packets in its queue in the following time slot. The non-zero elements of matrix $\mathbf M$ are given by
\begin{eqnarray}
m_{1,1}=&&\hspace{-6mm}P_S\;,\quad m_{1,2}=1-P_S\;,\quad m_{2,1}=P_S-q\;,\nonumber\\
 m_{2,3}=&&\hspace{-6mm}1-P_S  \;,\quad m_{L+1,L+1}=1-p \nonumber\\
m_{i,i+1}=&&\hspace{-6mm}1-p-q\;,\;\; m_{i+1,i}=p\;,\;\; m_{i,i}=q \;, \; {\rm for\;} i=1... L .\nonumber\\
\end{eqnarray}
Let $\mathbf{ {\rm Pr}\{Q\}}=[ {\rm Pr}\{Q=0  \},\;  {\rm Pr}\{Q=R_0  \},..., {\rm Pr}\{Q=LR_0  \}]$ denote the steady state probability vector of the considered Markov chain, where $ {\rm Pr}\{Q=kR_0  \}$, $k=0, \ldots, L$, is the probability of having $k$ packets
in the buffer. The steady state probability vector is obtained by solving the following system of equations 
\begin{eqnarray}\label{syst-eqa}
\left\{
\begin{array}{ccc}
\mathbf{ {\rm Pr}\{Q\}}\mathbf{M}&=&\mathbf{ {\rm Pr}\{Q\}}\\
\sum_{k=0}^{L} {\rm Pr}\{Q=kR_0  \}&=&1
\end{array}
    \right.\;,
\end{eqnarray}
which leads to (\ref{eq-Q-L-1}).
Using (\ref{eq-Q-L-1}) the average  queue size $E\{Q\}$ can be obtained from
\begin{eqnarray}
    E\{Q\}=R_0 \sum_{k=0}^L k {\rm Pr}\{Q=kR_0  \},
\end{eqnarray}
which leads to (\ref{mean-Q-1}). Furthermore, the average arrival rate can be found as
\begin{eqnarray}\label{arr-1}
   A=&&\hspace{-6mm}R_0\big[(1-P_S)  \big({\rm Pr}\{Q=0 \} +{\rm Pr}\{Q=R_0  \}\big) \nonumber\\
&&\hspace{-0mm} + (1-p-q) \big(1-{\rm Pr}\{Q=0 \} -{\rm Pr}\{Q=R_0  \}\nonumber\\
&&\hspace{-0mm} - {\rm Pr}\{Q=L R_0  \}\big) \big].
\end{eqnarray}
Inserting the average arrival rate given by (\ref{arr-1}) and the average queue size given by (\ref{mean-Q-1}) into (\ref{delay-main}) yields the average delay in (\ref{delay-1}).

For the case when $d_i$ is given by either (\ref{sol-d-delay-2}) or (\ref{sol-d-delay-3}), the queue in the buffer of the relay can be modeled by the Markov chain shown in Fig.~\ref{markov-chain-2}.  If the link selection 
variable $d_i$ is given by (\ref{sol-d-delay-2}), $p$ and $q$ are given by (\ref{p&q-1}), and if the link selection variable $d_i$ is given by (\ref{sol-d-delay-3}), $p$ and $q$ are given by (\ref{p&q-2}). 
Following the same procedure as before, (\ref{eq-Q-L-2})-(\ref{trup-delay-2}) can be obtained. This completes the proof.
\begin{figure}
\includegraphics[width=3.5in]{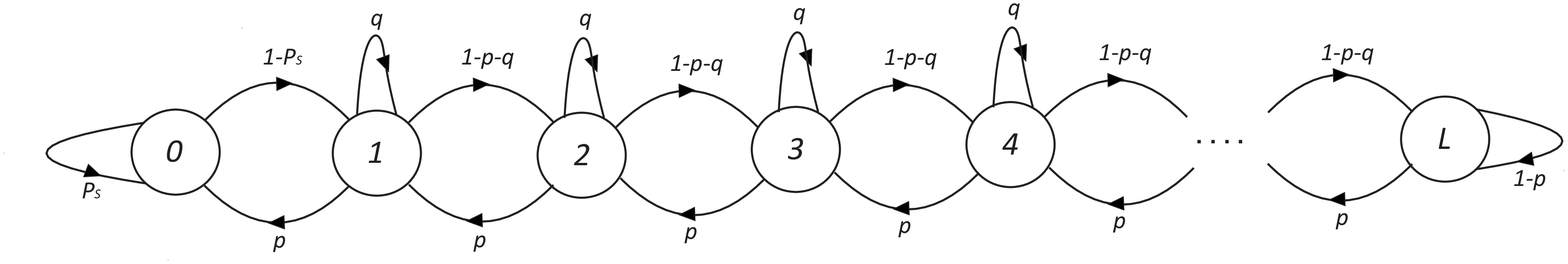}
\caption{Markov chain for the number of packets in the queue of the buffer if the link selection variable $d_i$ is given by $(\ref{sol-d-delay-2})$ or $(\ref{sol-d-delay-3})$.}\label{markov-chain-2}
\end{figure}

\subsection{Proof of Lemma \ref{lemma5}}\label{5a}
Let us first assume that $2p +q-1<0$, which is equivalent to $p<1-p-q$. Now, since $L\to \infty$,
$p^L$ goes to zero faster than $(1-p-q)^L$. Thus, by using $p^L=0$ as $L\to \infty$ in
 (\ref{delay-1}) and (\ref{delay-2}) , we obtain in both cases
\begin{eqnarray}
    E\{T\}=\frac{L}{p}-\frac{1}{1-2p-q}.
\end{eqnarray}
Thus, we conclude that if $2p +q-1<0$, $E\{T\}$ grows with $L$ and is unlimited as $L\to \infty$. Thus, if $E\{T\}$ is to be limited as $L\to \infty$, $2p +q-1>0$ has to hold. 

If $2p +q-1>0$, as $L\to \infty$, $(1-p-q)^L$ goes  to zero faster than $p^L$. Hence, (\ref{eq-Q_0-app-1})-(\ref{trup-delay-app-3}) are obtained by letting $(1-p-q)^L=0$, as $L\to \infty$, in the relevant equations in Theorem~\ref{theorem5} 
and inserting the corresponding $p$ and $q$ given by (\ref{p&q-1}) and (\ref{p&q-2}) into the resulting expressions. This concludes the proof.
\subsection{Proof of Lemma \ref{lemma5a}}\label{proof-lemma5a}
The minimum and maximum possible delays that the considered buffer-aided relaying system can achieve are obtained for $P_C=1$ and $P_C=0$, respectively. If $d_i$ is given by (\ref{sol-d-delay-1}), the delay is given by (\ref{delay-app-1}). 
By setting $P_C=1$ in (\ref{delay-app-1}) we obtain the minimum possible delay in (\ref{delay-app-min-1}). However, since (\ref{delay-app-1}) is valid only when $2p+q-1>0$, (\ref{delay-app-min-1}) is valid only when $P_R<1/(2-P_S)$. This condition
is obtained by inserting $P_C=1$ into the expressions for $p$ and $q$ given by (\ref{p&q-1}) and exploiting $2p+q-1>0$. On the other hand, in order to get the maximum delay given in  (\ref{delay-app-max-1}), we set $P_C=0$ in (\ref{delay-app-1}).  
The derived maximum delay is valid only when $P_S>1/(2-P_R)$, which is obtained from $2p+q-1>0$ and inserting $P_C=0$ into the expressions for $p$ and $q$ given by (\ref{p&q-1}).

A similar approach can be used to derive the delay limits $T_{\rm min,2}$, $T_{\rm max,2}$, $T_{\rm min,3}$, and $T_{\rm max,3}$ valid for the cases when $d_i$ is given by (\ref{sol-d-delay-2}) and (\ref{sol-d-delay-3}). 
This concludes the proof.
\subsection{Proof of Theorem \ref{theorem6}\label{app_J}}
The outage probability, $F_{\rm out}$, can be derived based on two different approaches. The first approach is straightforward and based on (\ref{OP-main}). However, the second approach provides more insight into how and when 
the outages occur and is based on counting the time slots in which no transmissions occur. In the following, we provide a proof based on the latter approach.

If $d_i$ is given by (\ref{sol-d-delay-1}) or (\ref{sol-d-delay-2}), there are four different cases where no node transmits.
\begin{enumerate}
\item The buffer is empty and the $\mathcal{S}$-$\mathcal{R}$ link is in outage.
\item The buffer in not empty nor full and both the $\mathcal{S}$-$\mathcal{R}$ and $\mathcal{R}$-$\mathcal{D}$ links are in outage.
\item The buffer is full and  the $\mathcal{S}$-$\mathcal{R}$  link is not in outage while the $\mathcal{R}$-$\mathcal{D}$ link is in outage. In this case, the source is selected for transmission but since the buffer is full, the packet is dropped.
\item  The buffer is full, both the $\mathcal{S}$-$\mathcal{R}$ and $\mathcal{R}$-$\mathcal{D}$ links are not in outage, and the source is selected for transmission based on the coin flip. In this case, since the buffer is full, the packet is dropped.
\end{enumerate}
Summing up the probabilities for each of the above four cases, we obtain (\ref{OP-delay-1-and-2}). 

If $d_i$ is given by (\ref{sol-d-delay-3}), an outage occurs in three cases: Case 1 and Case 2 as described above, and a new Case 3. In the new Case 3,  the buffer is full, the $\mathcal{S}$-$\mathcal{R}$ link is not in outage while the 
$\mathcal{R}$-$\mathcal{D}$ link is in outage, and the source is selected for transmission based on the coin flip. Summing up the probabilities for each of the three cases, we obtain (\ref{OP-delay-3}).
\subsection{Proof of Theorem \ref{theo-outage-high-snr}}\label{proof-outage-high-snr}
For delay constrained transmission with $E\{T\}<L$, the probability of dropped packets ${\rm Pr}\{Q=LR_0\}$ can be made arbitrarily small by increasing the buffer size $L$. Thus, for large enough $L$, we can set ${\rm Pr}\{Q=LR_0\}=0$ 
in (\ref{OP-delay-1-and-2}) and (\ref{OP-delay-3}).

In the high SNR regime, when $P_S\to 0$ and $P_R\to 0$, $P_R<1/(2-P_S)$ and $P_S<1/(2-P_R)$ always hold. Using $P_S\to 0$ and $P_R\to 0$ in the delays specified in Proposition \ref{proposition1}, we obtain the conditions $E\{T\}>3$
and $1<E\{T\}\leq 3$ if link selection variable $d_i$ is given by (\ref{sol-d-delay-1}) and (\ref{sol-d-delay-2}), respectively.

We first consider the case $E\{T\}>3$, where $d_i$ is given by (\ref{sol-d-delay-1}). Thus, the probability of the buffer being empty, ${\rm Pr}\{Q=0  \}$, is given by (\ref{eq-Q_0-app-1}). Using $P_S\to 0$ and $P_R\to 0$ in (\ref{eq-Q_0-app-1}), we obtain 
\begin{eqnarray}\label{eq-app-M-1}
    {\rm Pr}\{Q=0  \}=P_S\left(1-\frac{1}{2 P_C}\right).
\end{eqnarray} 
On the other hand, using $P_S\to 0$ and $P_R\to 0$ in the expression for $E\{T\}$ in (\ref{delay-app-1}), we obtain 
\begin{eqnarray}\label{eq-app-M-2}
  E\{T\}=\frac{1}{2 P_C -1}+2.
\end{eqnarray}
Solving (\ref{eq-app-M-2}) for $P_C$ yields  
\begin{eqnarray}\label{eq-app-M-3}
  P_C=\frac{1}{2}\left(1+\frac{1}{E\{T\}-2}\right).
\end{eqnarray}
Inserting (\ref{eq-app-M-3}) into (\ref{eq-app-M-1}) we obtain  
\begin{eqnarray}\label{eq-app-M-4}
    {\rm Pr}\{Q=0  \}=\frac{P_S}{E\{T\}-1}.
\end{eqnarray} 
Finally,  inserting (\ref{eq-app-M-4}) into (\ref{OP-delay-1-and-2}) and setting ${\rm Pr}\{Q=L R_0\}=0$, we obtain (\ref{OP-h-snr-1}).

Now, we consider the case $1<E\{T\}\leq 3$, where $d_i$ is given by (\ref{sol-d-delay-2}). Here, the probability of the buffer being empty, ${\rm Pr}\{Q=0  \}$, is given by (\ref{eq-Q_0-app-2}). For $P_S\to 0$ and $P_R\to 0$, we obtain from (\ref{eq-Q_0-app-2})
\begin{eqnarray}\label{eq-app-M-5}
    {\rm Pr}\{Q=0  \}=1-\frac{1}{2 {\rm  Pr}\{\mathcal{C}=1\}}.
\end{eqnarray} 
Furthermore, for $P_S\to 0$ and $P_R\to 0$, we obtain from (\ref{delay-app-2}) the asymptotic delay 
\begin{eqnarray}\label{eq-app-M-6}
  E\{T\}=\frac{1}{2 P_C -1}
\end{eqnarray}
or equivalently
\begin{eqnarray}\label{eq-app-M-7}
  P_C=\frac{1}{2}\left(1+\frac{1}{E\{T\}}\right).
\end{eqnarray}
Inserting (\ref{eq-app-M-7}) into (\ref{eq-app-M-5}) we obtain  
\begin{eqnarray}\label{eq-app-M-8}
    {\rm Pr}\{Q=0  \}=\frac{1}{E\{T\}+1}.
\end{eqnarray} 
Finally, inserting (\ref{eq-app-M-8}) into (\ref{OP-delay-3}) and setting ${\rm Pr}\{Q=L R_0  \}=0$, we obtain (\ref{OP-h-snr-2}). This concludes the proof.
\subsection{Proof of Theorem \ref{theorem7}\label{app_K}}
The Lagrangian of optimization problem (\ref{MPR-mixed-1}) is given by
\vspace{-2mm}
\begin{eqnarray}\label{MPR-mixed-2}
 \mathcal{L}=&&\hspace{-6mm}\frac{1}{N}\sum_{i=1}^N d_i \log_2\big(1+r(i)\big)  -\mu \frac{1}{N}\sum_{i=1}^N  \Big[d_i\log_2\big(1+r(i)\big) \nonumber\\
&&\hspace{-6mm} -(1-d_i) O_S(i) S_0\Big] - \sum_{i=1}^N \tilde{\beta}_i d_i(1-d_i),
\end{eqnarray}
where $\mu$ and $\tilde{\beta}_i$ are the Lagrange multipliers. By differentiating $\mathcal{L}$ with respect to  $d_i$, introducing $\beta_i=N\tilde{\beta}_i$, equating the result to zero, and solving the  equation with respect to   $d_i$,  we obtain 
\begin{eqnarray}\label{link-select-mixed-non-PA}
   d_i=\left\{
\begin{array}{cl}
1 & \textrm{if }  (1-\mu) \log_2\big(1+r(i)\big)\geq \mu O_S(i) S_0\\
0 & \textrm{if }  (1-\mu) \log_2\big(1+r(i)\big) \leq \mu O_S(i) S_0 ,
\end{array} 
\right.
\end{eqnarray}
where we took into account that  $\beta_i<0$. Since for $\mu < 0$ and $\mu > 1$, we have always $d_i=1$ and $d_i=0$, respectively, irrespective of the (non-negative) values of $\log_2\big(1+r(i)\big)$  and $O_S(i) S_0$,
$0\leq \mu \leq 1$ has to hold. 

Let us first consider the case $0 < \mu < 1$ and investigate the boundary values $\mu=0$ and $\mu=1$ later. For $0 < \mu < 1$, (\ref{link-select-mixed-non-PA}) can be written in the form of (\ref{sol-d-mixed-1}) after setting $\rho=\mu/(1-\mu)$, 
where $\rho$ is chosen such that constraint C1 of problem (\ref{MPR-mixed-1}) is met. Denoting the pdfs of $s(i)$ and $r(i)$ by $f_s(s)$ and $f_r(r)$ constraint C1 of problem (\ref{MPR-mixed-1}) can be rewritten as in (\ref{eq11}), which is valid 
for $\rho$ in the range of $\rho=[0,\infty)$. Thus, by setting $\rho=\infty$ in (\ref{eq11}), we obtain the entire domain over which  (\ref{sol-d-mixed-1}) is valid, which leads to condition (\ref{cond-PS-mixed-non-pa-1}). 
 

Next, we consider the boundary values $\mu=0$ and $\mu=1$. The boundary value $\mu=0$ or equivalently $\rho=0$ is relevant only in the trivial case when the $\mathcal{S}$-$\mathcal{R}$ link is never in outage (i.e. ~$P_S=0$) and $S_0=\infty$, where a trivial  solution is 
given by $d_1=0$ and $d_i=1$ for $i=2,\ldots,N$ and $N\to \infty$. 

The other boundary value,  $\mu=1$, is invoked only when by using $d_i$ as defined in  (\ref{sol-d-mixed-1}), constraint C1 cannot be satisfied even when $\rho\to \infty$, which is the case when condition (\ref{cond-PS-mixed-non-pa-1}) does not hold. 
Therefore, if (\ref{cond-PS-mixed-non-pa-1}) does not hold, we set  $\mu=1$ in (\ref{link-select-mixed-non-PA}) and obtain the following cases:
\begin{enumerate}
\item If $O_S(i)=1$, then $d_i=0$.
\item If $O_S(i)=0$, then $d_i$ can be chosen arbitrarily to be either zero or one as long as constraint C1 holds. 
\end{enumerate}
However, the same throughput as obtained when $O_S(i)=0$ and $d_i$ is chosen such that constraint C1 holds, can also be obtained  by choosing always  $d_i=1$ when $O_S(i)=0$ resulting in (\ref{sol-d-mixed-2}). 
The reason behind this is as follows: Assume there is a policy for which when $O_S(i)=0$, $d_i$ is chosen such that  constraint C1 holds. Now, we change $d_i$ from $0$ to $1$ for  $O_S(i)=0$. However, this change does not affect
the (average) amount of data entering the buffer. Thus, because of the law of conservation of flow, the average amount of data entering the buffer per time slot  is identical to the average amount of data leaving the buffer per time slot (the throughput),  
and the throughput is not affected by the change. 
%

\subsection{Proof of Theorem \ref{theorem8}\label{app_L}}
The Lagrangian of optimization problem (\ref{MPR-mixed-1-PA}) is given by
\begin{eqnarray}\label{MPR2a}
    \mathcal{L}=&&\hspace{-6mm}\frac{1}{N}\sum_{i=1}^N d_i \log_2(1+\gamma_R(i) h_R(i)) -\sum_{i=1}^N \tilde{\beta}_i d_i(1-d_i)\nonumber\\
-&&\hspace{-6mm}
\mu \frac{1}{N}\sum_{i=1}^N \Big[d_i \log_2(1+\gamma_R(i) h_R(i))-(1-d_i) O_S(i) S_0 \Big]\nonumber\\
&&- \nu \frac{1}{N}\sum_{i=1}^N \Big[ (1-d_i) O_S(i) \gamma_S+d_i \gamma_R(i)\Big],
\end{eqnarray}
where the Lagrange multipliers $\mu$, $\tilde\beta_i$, and $\nu$ are chosen such that C1, C2, and C3 are satisfied, respectively.
We again consider only the interval $0\leq \mu \leq 1$ as for $\mu < 0$ and $\mu > 1$, we have always $d_i=1$ and $d_i=0$, respectively, 
irrespective of the (nonnegative) values of $\log_2\big(1+r(i)\big)$ and $O_S(i) S_0$. 

We concentrate first on the case $0< \mu < 1$ and consider the boundary values later. By differentiating $\mathcal{L}$ with respect to $\gamma_R(i)$ and $d_i$, introducing $\beta_i=N\tilde{\beta}_i$, and setting 
the results to zero, we obtain two equations. Solving the resulting system of equations with respect to $\gamma_R(i)$ and $d_i$, and taking into account that  $\beta_i<0$, 
$0<\mu<1$, and $\nu>0$, we obtain  (\ref{power-eq-2b}) and (\ref{sol-d-mixed-PA}) after letting $\rho=\ln(2) \mu/(1-\mu)$ and $\lambda=\ln(2) \nu/(1-\mu)$, which are chosen 
such that constraints C1 and  C3 are met with equality. 
Given the pdfs $f_{h_S}(h_S)$ and $f_{h_R}(h_R)$, conditions (\ref{cond-1-mixed}) and (\ref{cond-mixed-2-PA}) 
can be directly written as (\ref{nz-1}) and (\ref{nz-1a}), respectively.  Setting $\rho\to \infty$ in (\ref{nz-1}) and (\ref{nz-1a}), 
we obtain condition (\ref{cond-PS-mixed-pa-1}) which is necessary for the validity of (\ref{sol-d-mixed-2}).

Similar to the fixed transmit power case, the boundary value $\mu=0$ is trivial. On the other hand, for $\mu=1$, we obtain that $d_i$ has to be set to $d_i=0$ when $O_S(i)=1$ and 
for $O_S(i)=0$, $d_i$ can be chosen arbitrarily. Similar to the fixed power case, we set $d_i=1$ when $O_S(i)=1$ in order to minimize the delay.
Thus, the optimal power and link selection variables are given by  (\ref{power-eq-2b-2}) and (\ref{sol-d-mixed-PA-2}), respectively, and the throughput is given by (\ref{max-tau-mixed-pa-2}).



\subsection{Proof of Theorem \ref{theorem9}\label{app_theorem9}}
For $\gamma_S=\gamma_R=\gamma\to\infty$, the protocol in Proposition~\ref{pr_3} is optimal in the sense that it maximizes the throughput while satisfying the average delay constraint. In particular, for high SNR in the $\mathcal{S}$-$\mathcal{R}$ link, the probability that the link
is in outage approaches zero and the relay receives $S_0$ bits per source transmission. On the other hand, the number of bits transmitted by the relay in one time slot over the $\mathcal{R}$-$\mathcal{D}$ link increases with the SNR. Thus, for sufficiently high SNR, 
the source transmits $kS_0$ bits in $k$ time slots and the relay needs just $n=1$ time slot to forward the entire information to the destination. Hence, every transmission period comprises $k+n=k+1$ time slots, where the queue length at the relay increases from $S_0$ to $kS_0$ in
the first $k$ time slots and is reduced to zero in the $(k+1)$th time slot. Hence, the average queue length, $E\{Q\}$, can be written as
\begin{eqnarray}\label{e-q-n-1-1}
    E\{Q\}\to&&\hspace{-6mm}\frac{1}{k+1}(1+2+...+k+0) S_0= \frac{1}{k+1}\frac{k(k+1)}{2} S_0\nonumber\\
=&&\hspace{-6mm} \frac{k}{2}S_0 
,\quad \textrm{as }\gamma\to\infty\;.
\end{eqnarray}
On the other hand, the arrival rate is identical to the throughput and given by (\ref{eq_m_t_1}), and for high SNR it converges to
\begin{eqnarray}\label{e-q-n-1-2}
 A=\tau\to S_0 \frac{k}{k+1}\label{r_tt1_2}\;,\quad \textrm{as }\gamma\to\infty\;.
\end{eqnarray}
Combining (\ref{delay-main}), (\ref{e-q-n-1-1}), and (\ref{e-q-n-1-2}) the average delay is found as
\begin{eqnarray}
    E\{T\}&\to&\frac{k+1}{2}\label{r_tt1_1}\;,\quad \textrm{as }\gamma\to\infty\;.
\end{eqnarray}
Finally, combining (\ref{r_tt1_2}) and (\ref{r_tt1_1})  the throughput  can be expressed as (\ref{t_b_N-1_2}), and the multiplexing gain in (\ref{t_b_N-1_22121})
follows directly.
\end{appendix}


\bibliography{litdab}
\bibliographystyle{IEEETran}

\begin{IEEEbiographynophoto}{Nikola Zlatanov}
(S'06) was born in Macedonia. He received the Dipl.Ing. and M.S. degrees in electrical
engineering from Sts. Cyril and Methodius University, Skopje, Macedonia, in
2007 and 2010, respectively. Currently, he is working towards the Ph.D. degree
at the University of British Columbia (UBC), Vancouver, BC, Canada.
His current research interests include the general field of wireless communications  with emphasis on   buffer-aided relaying.

Mr. Zlatanov received the Four Year Doctoral Fellowship by UBC in 2010. In
2011, he received the UBC's  Killam  Doctoral Scholarship  and was awarded
Young Scientist of the Year   by the President of the Republic of Macedonia. In 2012, he received the Vanier Canada Graduate Scholarship.
\end{IEEEbiographynophoto}

\begin{IEEEbiographynophoto}{Robert Schober}
 (S'98, M'01, SM'08, F'10) was born in Neuendettelsau, Germany, in 1971. He received the Diplom (Univ.) and the Ph.D. degrees in electrical engineering from the University of Erlangen-Nuermberg in 1997 and 2000, respectively. From May 2001 to April 2002 he was a Postdoctoral Fellow at the University of Toronto, Canada, sponsored by the German Academic Exchange Service (DAAD). Since May 2002 he has been with the University of British Columbia (UBC), Vancouver, Canada, where he is now a Full Professor. Since January 2012 he is an Alexander von Humboldt Professor and the Chair for Digital Communication at the Friedrich Alexander University (FAU), Erlangen, Germany. His research interests fall into the broad areas of Communication Theory, Wireless Communications, and Statistical Signal Processing.

Dr. Schober received several awards for his work including the 2002 Heinz Maier-Leibnitz Award of the German Science Foundation (DFG), the 2004 Innovations Award of the Vodafone Foundation for Research in Mobile Communications, the 2006 UBC Killam Research Prize, the 2007 Wilhelm Friedrich Bessel Research Award of the Alexander von Humboldt Foundation, the 2008 Charles McDowell Award for Excellence in Research from UBC, a 2011 Alexander von Humboldt Professorship, and a 2012 NSERC E.W.R. Steacie Fellowship. In addition, he received best paper awards from the German Information Technology Society (ITG), the European Association for Signal, Speech and Image Processing (EURASIP), IEEE WCNC 2012, IEEE Globecom 2011, IEEE ICUWB 2006, the International Zurich Seminar on Broadband Communications, and European Wireless 2000. Dr. Schober is a Fellow of the Canadian Academy of Engineering and a Fellow of the Engineering Institute of Canada. He is currently the Editor-in-Chief of the IEEE Transactions on Communications.
\end{IEEEbiographynophoto}

\end{document}